\documentclass[12pt, a4paper, reqno]{amsart}
\usepackage{mathrsfs}
\usepackage{bbm}
\usepackage{amsmath,amscd,amssymb,latexsym}
\usepackage[all]{xy}
\usepackage{color}
\usepackage[section]{placeins}
\usepackage[latin1]{inputenc}

\usepackage{algorithm}
\usepackage{algpseudocode}
\usepackage{amsmath}
\usepackage{cases}
\usepackage{arydshln}
\usepackage{cite}



\input xypic

\evensidemargin=\oddsidemargin
\allowdisplaybreaks[4]

\DeclareMathVersion{can}
\DeclareMathAlphabet{\can}{OT1}{cmss}{m}{n}
\newtheorem{thm}{Theorem}[section]
\newtheorem{cor}[thm]{Corollary}

\newtheorem{lem}[thm]{Lemma}
\newtheorem{prop}[thm]{Proposition}
\newtheorem{rem}[thm]{Remark}
\newtheorem{exa}[thm]{Example}
\theoremstyle{definition}
\newtheorem{defn}[thm]{Definition}
\theoremstyle{fact}

\theoremstyle{conjecture}

\numberwithin{equation}{section}

\newcommand{\ord}{\operatorname{ord}}

\begin{document}

\title[Generator polynomials]{Generator polynomials of cyclic expurgated or extended Goppa codes}

\author [Jia] {Xue Jia}
\address{\rm School of Mathematics, Nanjing University of Aeronautics and Astronautics, Nanjing,  211100, P. R. China and State Key Laboratory of Cryptology, P.O. Box 5159, Beijing 100878, P. R. China}
\email{jiaxue0904@163.com}

\author[Li]{Fengwei Li}
\address{\rm
 College of Science, Nanjing University of Posts and Telecommunications, Nanjing, 210023, P. R. China} \email{lfwzzu@126.com}

\author [Sun] {Huan Sun}
\address{\rm School of Mathematics, Nanjing University of Aeronautics and Astronautics, Nanjing,  211100, P. R. China and State Key Laboratory of Cryptology, P.O. Box 5159, Beijing 100878, P. R. China}
\email{sunhuan6558@163.com}
\author[Yue]{Qin Yue}
\address{\rm School of Mathematics, Nanjing University of Aeronautics and Astronautics, Nanjing,  211100, P. R. China and State Key Laboratory of Cryptology, P.O. Box 5159, Beijing 100878, P. R. China}
\email{yueqin@nuaa.edu.cn}

\thanks{The paper is supported by the National Natural Science Foundation of China ( No. 62172219 and No.
 12171420), the Natural Science Foundation of Shandong Province under Grant ZR2021MA046, the Natural Science Foundation of Jiangsu Province under Grant BK20200268, the Innovation Program for Quantum Science and Technology under Grant 2021ZD0302902.\\
(*Corresponding authours: Fengwei Li and Qin Yue)}

 \keywords{Cyclic codes, Expurgated Goppa codes, Extended Goppa codes, Generator polynomials}

\subjclass[2010]{94B05}


\begin{abstract}
Classical Goppa codes are a well-known class of codes with applications in code-based cryptography, which are a special case of alternant codes. Many papers are devoted to the search for Goppa codes with a cyclic extension or with a cyclic parity-check subcode.
Let $\Bbb F_q$ be a finite field with $q=2^l$ elements, where  $l$ is a positive integer. In this paper, we determine all the generator polynomials of cyclic expurgated or extended Goppa codes under some prescribed permutations induced by the projective general linear automorphism $A \in PGL_2(\Bbb F_q)$.  Moreover, we provide some examples to support our findings.
\end{abstract}
\maketitle
\section{Introduction}
In coding theory, cyclic codes \cite{CZ1,HP,XL} are an interesting type of linear codes and have wide applications in communication and storage systems due to their efficient encoding and decoding algorithms. Let $\Bbb F_q$ be a finite field with $q$ elements, where $q=p^s$, $p$ is a prime, and $s$ is a positive integer. An $[n,k,d]$ linear code $\mathcal C$ is a $k$-dimensional subspace of $\Bbb F_q^n$ with minimum distance $d$. It is called cyclic if $(c_0,c_1,\ldots,c_{n-1}) \in \mathcal C$ implies $(c_{n-1},c_0,c_1,\ldots,c_{n-2}) \in \mathcal C$. By identifying the codeword $(c_0,c_1,\ldots,c_{n-1}) \in \Bbb F_q^n$ with the polynomial
$$
c_0+c_1x+c_2x^2+\cdots+c_{n-1}x^{n-1} \in \Bbb F_q[x]/(x^n-1),
$$
a code $\mathcal C$ of length $n$ over $\Bbb F_q$ corresponds to a subset of $\Bbb F_q[x]/(x^n-1)$. Then $\mathcal C$ is cyclic if and only if the corresponding subset is an ideal of $\Bbb F_q[x]/(x^n-1)$. Note that every ideal of $\Bbb F_q[x]/(x^n-1)$ is principal. Hence $\mathcal C=\langle u(x) \rangle$, where $u(x)$ is a monic divisor of $x^n-1$ over $\Bbb F_q$. Then $u(x)$ is called the generator polynomial and $h(x)=(x^n-1)/u(x)$ is called the parity-check polynomial of the cyclic code $\mathcal C$. A cyclic code is totally determined by its generator polynomial, so it is significant to investigate the generator polynomial of a cyclic code.

Goppa codes \cite{GV} are subfield subcodes of generalized Reed-Solomon (GRS) codes and a special case of alternant codes \cite{HP,MS,XL}. In recent decades, Goppa codes have received extensive attention in the field of coding and cryptography.
In 1978 \cite{MC}, McEliece introduced the first code-based cryptosystem, in which he used a random binary irreducible Goppa code as the private key. Now the principle of the McEliece system still plays an important role in the design of code-based cryptography. In particular, the Classic McEliece originated from McEliece is one of the four candidate algorithms for the fourth round of the National Institute of Standards and Technology (NIST) post-quantum cryptography competition, and is considered to be very promising to be standardized by NIST \cite{P}.
In terms of structural cryptanalysis, there are several effective attacks on families associated with random irreducible Goppa codes \cite{Si,WC,CA,FJ,BR}. Moreover, there exist some key-recovery approaches that work on unstructured Goppa codes, but whose complexity is exponential at least in the more general framework \cite{K,CR}. It can be seen that the key using Goppa codes is sufficiently secure.
 Recall that there are two general constructions in coding theory that allow to define new codes from a given code, namely expurgated and extended codes, which satisfy certain linear constraints and are designed to increase code redundancy.
 In summary, it is very meaningful to study the structures and properties of these algebraic codes.


It is well known that some of Goppa codes have a cyclic extension (cf. \cite{MS}, Research Problem 12.3). Many papers are devoted to the search for such families of Goppa codes \cite{BM},\cite{FT},\cite{S,T,TY,TZ,V}. In addition, cyclic Goppa codes, cyclic expurgated and extended Goppa codes were also investigated in \cite{Ba,Bb,Bc,Bd,Be} and \cite{LY}. Specifically, in \cite{Bd}, Berger proved that the parity-check subcodes of Goppa codes and the extended Goppa codes are both alternant codes and gave new families of Goppa codes with a cyclic extension, and some families of non-cyclic Goppa codes with a cyclic parity-check subcode.
 Berger presented the construction of Goppa codes invariant under some prescribed permutation \cite{Bc}. Moreover, Berger constructed some other families of Goppa codes with a cyclic extension and gave a necessary and sufficient condition for a Goppa code to have a cyclic extension \cite{Bb}.
 However, it is not yet known the generator polynomials of these cyclic codes.
 Based on the above results, we shall determine all the generator polynomials of cyclic expurgated or extended Goppa codes under some prescribed permutations induced by the projective general linear automorphism $A \in PGL_2(\Bbb F_q)$.

This paper is organized as follows. In section 2, we recall some definitions of Goppa codes, expurgated and extended Goppa codes. Besides, we recall the method to construct cyclic expurgated or extended Goppa codes. In section 3, we show generator polynomials of cyclic expurgated or extended Goppa codes and provide some examples to support our findings. In section 4, we conclude the paper.
\section{Preliminaries}
In this paper, we always assume that $q=2^l$, $\Bbb F_q$ denotes a finite field with $2^l$ elements, $\overline{\Bbb F}_q=\Bbb F_q \bigcup \{\infty\}$ is a set of coordinates for the projective line, and $\Bbb F_q[x]$ is the polynomial ring over $\Bbb F_q$. We shall review some basic notations and necessary preliminaries.
\subsection{\small{Goppa Codes, Expurgated and Extended Goppa codes}}\

In this subsection, we shall recall the definitions of Goppa codes, expurgated and extended Goppa codes, which can be found in \cite{Bc,Bd,HP,MS,R,XL}.
\begin{defn}\label{def 2.1}
Let $\mathcal L=(\alpha_1,\ldots, \alpha_n)$ be an $n$-tuple of distinct elements of $\Bbb F_q$ and $g(x)=\sum\limits_{i=0}^r g_ix^i \in \Bbb F_q[x]$ a polynomial of degree $r(<n)$ such that $g(\alpha_i) \neq 0$ for $i=1, \ldots, n$. The Goppa code $\Gamma(\mathcal L, g)$ with the support $\mathcal L$ and the Goppa polynomial $g(x)$ is defined as
$$
\Gamma(\mathcal L, g)=\left\{\boldsymbol c=(c_1, c_2, \ldots, c_n) \in \Bbb F_2^n: \sum_{i=1}^{n} \frac{c_i}{x-\alpha_i} \equiv 0 \pmod {g(x)}\right\}.
$$
The expurgated Goppa code $\widetilde{\Gamma}(\mathcal L, g)$ of $\Gamma(\mathcal L, g)$ is defined as
$$
\widetilde{\Gamma}(\mathcal L, g)=\left\{\boldsymbol c=(c_1, c_2, \ldots, c_n) \in \Gamma(\mathcal L, g): \sum_{i=1}^n c_i=0\right\}.
$$
The extended Goppa code $\overline{\Gamma}(\mathcal L, g)$ is defined as
$$
\overline{\Gamma}(\mathcal L, g)=\left\{\boldsymbol c=(c_1, \ldots, c_n, c_{n+1}) \in \Bbb F_2^{n+1}: (c_1, c_2, \ldots, c_n) \in \Gamma(\mathcal L, g), \sum_{i=1}^{n+1} c_i=0\right\}.
$$

Moreover, the parity-check matrices of $\Gamma(\mathcal L, g), \widetilde{\Gamma}(\mathcal L, g)$, and $\overline{\Gamma}(\mathcal L, g)$ are given by
$$
H=\left(\begin{array}{cccc}
g(\alpha_1)^{-1} & g(\alpha_2)^{-1} & \cdots & g(\alpha_n)^{-1} \\
\alpha_1g(\alpha_1)^{-1} & \alpha_2g(\alpha_2)^{-1} & \cdots & \alpha_ng(\alpha_n)^{-1} \\
\vdots & \vdots & & \vdots \\
\alpha_1^{r-1}g(\alpha_1)^{-1} & \alpha_2^{r-1}g(\alpha_2)^{-1} & \cdots & \alpha_n^{r-1}g(\alpha_n)^{-1}
\end{array}\right),
$$
\begin{equation} \label{widetilde H}
\widetilde H=\left(\begin{array}{cccc}
g(\alpha_1)^{-1} & g(\alpha_2)^{-1} & \cdots & g(\alpha_n)^{-1} \\
\alpha_1g(\alpha_1)^{-1} & \alpha_2g(\alpha_2)^{-1} & \cdots & \alpha_ng(\alpha_n)^{-1} \\
\vdots & \vdots & & \vdots \\
\alpha_1^r g(\alpha_1)^{-1} & \alpha_2^r g(\alpha_2)^{-1} & \cdots & \alpha_n^r g(\alpha_n)^{-1}
\end{array}\right),
\end{equation}
\begin{equation}
\overline H=\left(\begin{array}{ccccc}
g(\alpha_1)^{-1} & g(\alpha_2)^{-1} & \cdots & g(\alpha_n)^{-1} & 0\\
\alpha_1g(\alpha_1)^{-1} & \alpha_2g(\alpha_2)^{-1} & \cdots & \alpha_ng(\alpha_n)^{-1} & 0\\
\vdots & \vdots & & \vdots & \vdots\\
\alpha_1^{r-1}g(\alpha_1)^{-1} & \alpha_2^{r-1}g(\alpha_2)^{-1} & \cdots & \alpha_n^{r-1}g(\alpha_n)^{-1} & 0 \\
\alpha_1^r g(\alpha_1)^{-1} & \alpha_2^r g(\alpha_2)^{-1} & \cdots & \alpha_n^r g(\alpha_n)^{-1} & g(\infty)^{-1}
\end{array}\right),
\end{equation}
where $g(\infty)=g_r$.
If $g(x)$ is irreducible over $\Bbb F_q$, then $\Gamma(\mathcal L, g)$, $\widetilde{\Gamma}(\mathcal L, g)$, and $\overline{\Gamma}(\mathcal L, g)$ are called irreducible.
\end{defn}
\begin{lem}(\cite{MS}) \label{lem Goppa}
Let $\Gamma(\mathcal L, g)$ be a binary Goppa code in Definition \ref{def 2.1}. Suppose that the Goppa polynomial $g(x) \in \Bbb F_q[x]$ has no multiple zeros, then $\Gamma(\mathcal L, g)=\Gamma(\mathcal L, g^2)$.
\end{lem}

\subsection{\small{Action of Groups}}\

In this subsection, we shall recall the notations of some matrix groups, the actions of the projective linear group and the projective semi-linear group on $\overline{\Bbb F}_q=\Bbb F_q \bigcup \{\infty\}$.

There are some matrix groups as follows:

(1) The affine group
$$
AGL_2(\Bbb F_q)=\left\{A=\left(\begin{array}{cc}
a & b \\
0 & 1
\end{array}\right): a \in \Bbb F_q^*, b \in \Bbb F_q\right\}.
$$

(2) The general linear group
$$
GL_2(\Bbb F_q)=\left\{A=\left(\begin{array}{cc}
a & b \\
c & d
\end{array}\right): a, b, c, d \in \Bbb F_q, ad-bc \neq 0\right\}.
$$

(3) The projective general linear group
$$
PGL_2(\Bbb F_q)=GL_2(\Bbb F_q)/\left\{k E_2: k \in \Bbb F_q^*\right\},
$$
where $E_2$ is the $2\times 2$ identity matrix.

(4) The projective semi-linear group
$$
P\Gamma L_2(\Bbb F_q)=PGL_2(\Bbb F_q)\times Gal(\Bbb F_q/\Bbb F_2),
$$
where $Gal(\Bbb F_q/\Bbb F_2)=\langle\sigma\rangle$ is the Galois group.

Let $\overline{\Bbb F}_q=\Bbb F_q \bigcup \{\infty\}$ be the projective line set and $A=\left(\begin{array}{cc}
a & b \\
c & d
\end{array}\right)\in GL_2(\Bbb F_q)$. Then the projective general linear group $PGL_2(\Bbb F_q)$ acts on $\overline{\Bbb F}_q$ as follows:
$$
\begin{aligned}
 PGL_2(\Bbb F_q) \times \overline{\Bbb F}_q &\rightarrow \overline{\Bbb F}_q \\
\widetilde A (\zeta) &\mapsto \widetilde A(\zeta)=\frac{a \zeta+b}{c \zeta+d}.
\end{aligned}
$$
The projective semi-linear group $P\Gamma L_2(\Bbb F_q)$ acts on $\overline{\Bbb F}_q$ as follows:
$$
\begin{aligned}
 P\Gamma L_2(\Bbb F_q) \times \overline{\Bbb F}_q &\rightarrow \overline{\Bbb F}_q \\
(\widetilde A,\sigma^j) (\zeta) &\mapsto \widetilde A(\sigma^j(\zeta))=\frac{a \zeta^{2^j}+b}{c \zeta^{2^j}+d},
\end{aligned}
$$
where $\frac10=\infty, \frac1\infty=0$ and $\frac{a\infty+b}{c\infty+d}=\frac ac.$

Let $M=(A, \sigma^j)\in P\Gamma L_2(\Bbb F_q)$ be of order $n$ and $G=\langle M\rangle $  a cyclic subgroup of $P\Gamma L_2(\Bbb F_q)$, then for $\alpha\in \overline {\Bbb F}_q$,
$$O_{\alpha}=G(\alpha)=\{M^i(\alpha): 0\le i\le n-1\}$$
is called an orbit of $\alpha$ under the action of $G$.

\subsection{\small{Cyclic Expurgated or Extended Goppa Codes}}\

In this subsection, we shall recall the method to construct cyclic expurgated or extended Goppa codes. In fact, we can    choose some special Goppa polynomials and support sets to obtain cyclic expurgated or extended Goppa codes. There are many articles (\cite{Bb,Bd,Be}) describing such results.
\begin{lem}\label{lem cyclic}
 Let $M=\left(\left(\begin{array}{cc}
a & b \\
c & d
\end{array}\right),\sigma^j\right)\in P\Gamma L_2(\Bbb F_q)$ with $c\ne 0$, $G=\langle M\rangle$ a subgroup of $P\Gamma L_2(\Bbb F_q)$, $O_{\alpha}=\{\alpha_1, \ldots, \alpha_n\}$  an orbit of $\alpha\in\overline {\Bbb F}_q$ under the action of $G$, and $g(x)=\sum\limits_{i=0}^r g_i x^i \in \Bbb F_q[x]$  a monic polynomial of degree $r(<n)$ without roots in $O_{\alpha}$ satisfying

\begin{equation}\label{eq 2.1}
c^r g(\frac ac) \neq 0 ~~\text{and} ~~
(cx^{2^j}+d)^r g\left(\frac{ax^{2^j}+b}{cx^{2^j}+d}\right)=c^r g(\frac ac)g(x)^{2^j}.
\end{equation}

(1) If $\infty \notin\mathcal L=O_{\alpha}$, then the expurgated Goppa code $\widetilde{\Gamma}(\mathcal L, g)$ is cyclic of length $n$.

(2) If $\infty \in O_{\alpha}$,  $\mathcal L=O_{\alpha} \backslash \{\infty\}$, then the extended Goppa code $\overline{\Gamma}(\mathcal L, g)$ is cyclic of length $n$.
\end{lem}
\begin{proof}
(1) Since $O_{\alpha}=\{\alpha_1, \ldots, \alpha_n\}$ is an orbit of $\alpha\in\overline {\Bbb F}_q$ under the action of $G$,
$$\alpha_{i+1}=M \alpha_i=\frac{a\alpha_i^{2^j}+b}{c\alpha_i^{2^j}+d}, i=1,\ldots, n,$$
where $\alpha_{n+1}=\alpha_1$.

If $\widetilde H$ is  the parity-check matrix of $\widetilde{\Gamma}(\mathcal L, g)$ in (\ref{widetilde H}), where $\infty \notin\mathcal L=O_{\alpha}$,  then by (2.3),
\begin{eqnarray*}
\widetilde H&=&\left(\begin{array}{cccc}
g(M\alpha_{n})^{-1} & g(M\alpha_1)^{-1} & \cdots & g(M\alpha_{n-1})^{-1} \\
M\alpha_{n}g(M\alpha_n)^{-1} & M\alpha_1g(M\alpha_1)^{-1} & \cdots & M\alpha_{n-1}g(M\alpha_{n-1})^{-1} \\
\vdots & \vdots & & \vdots \\
(M\alpha_n)^r g(M\alpha_n)^{-1} & (M\alpha_1)^r g(M\alpha_1)^{-1} & \cdots & (M\alpha_{n-1})^r g(M\alpha_{n-1})^{-1}
\end{array}\right)\\
&=&\frac 1{c^r g(\frac ac)}\left(\begin{array}{cccc}
\frac{(c\alpha_n^{2^j}+d)^r}{g(\alpha_n)^{2^j}} & \frac{(c\alpha_1^{2^j}+d)^r}{g(\alpha_1)^{2^j}} & \cdots & \frac{(c\alpha_{n-1}^{2^j}+d)^r}{g(\alpha_{n-1})^{2^j}} \\
\frac{(a\alpha_{n}^{2^j}+b)(c\alpha_{n}^{2^j}+d)^{r-1}}{g(\alpha_{n})^{2^j}} & \frac{(a\alpha_{1}^{2^j}+b)(c\alpha_{1}^{2^j}+d)^{r-1}}{g(\alpha_{1})^{2^j}} & \cdots & \frac{(a\alpha_{n-1}^{2^j}+b)(c\alpha_{n-1}^{2^j}+d)^{r-1}}{g(\alpha_{n-1})^{2^j}}  \\
\vdots & \vdots & & \vdots \\
\frac{(a\alpha_n^{2^j}+b)^r}{g(\alpha_n)^{2^j}} & \frac{(a\alpha_1^{2^j}+b)^r}{g(\alpha_1)^{2^j}} & \cdots & \frac{(a\alpha_{n-1}^{2^j}+b)^r}{g(\alpha_{n-1})^{2^j}}
\end{array}\right).
\end{eqnarray*}
By some elementary row transformations, $\widetilde H$ can be changed into
$$\widetilde H_1=\left(\begin{array}{cccc}
\frac{1}{g(\alpha_n)^{2^j}} & \frac{1}{g(\alpha_1)^{2^j}} & \cdots & \frac{1}{g(\alpha_{n-1})^{2^j}} \\
\frac{\alpha_{n}^{2^j}}{g(\alpha_{n})^{2^j}} & \frac{\alpha_{1}^{2^j}}{g(\alpha_{1})^{2^j}} & \cdots & \frac{\alpha_{n-1}^{2^j}}{g(\alpha_{n-1})^{2^j}}  \\
\vdots & \vdots & & \vdots \\
\frac{(\alpha_n^r)^{2^j}}{g(\alpha_n)^{2^j}} & \frac{(\alpha_1^r)^{2^j}}{g(\alpha_1)^{2^j}} & \cdots & \frac{(\alpha_{n-1}^r)^{2^j}}{g(\alpha_{n-1})^{2^j}}
\end{array}\right).
$$
Then
$\boldsymbol x=(x_1,\ldots,x_n)\in \widetilde{\Gamma}(\mathcal L, g) \subseteq \Bbb F_2^n$ if and only if $\widetilde H \boldsymbol x^\mathrm T=0$ if and only if $$
\widetilde H_1 (x_n, x_1,\ldots, x_{n-1})^\mathrm T=0$$
if and only if $(x_n, x_1,\ldots, x_{n-1})\in \widetilde{\Gamma}(\mathcal L, g)$.
Hence $\widetilde{\Gamma}(\mathcal L, g)$ is cyclic.

(2) For $\infty \in O_{\alpha}$, without loss of generality, let $\alpha_n=\infty$
and  $O_{\alpha}=\{\alpha_1, \ldots, \alpha_{n-1},\alpha_n=\infty\}$.
Similarly, we know that $\overline {\Gamma}(\mathcal L, g)$ is also cyclic.
\end{proof}

\begin{lem}\label{lem g(x)} (\cite[Theorem 6, Proposition 7]{Bd})

(1) Let $M=(A, \sigma^j) \in P\Gamma L_2(\Bbb F_q)$, $\alpha$ be an element of an extension of $\Bbb F_q$, and $O_{\alpha}(x)=\prod\limits_{i=0}^{n-1}(x-M^i(\alpha))$, where $n$ is the length of the orbit of $\alpha$ under $G=\langle M\rangle$. Let $g(x)=\sum\limits_{i=0}^r g_i x^i \in \Bbb F_q[x]$ be a monic polynomial of degree $r$. Then $g(x)$ satisfies the condition in (\ref{eq 2.1}) if and only if $g(x)$ is a product of polynomials $O_{\alpha_k}(x)$ for some $\alpha_k$.

(2) Let $M \in P\Gamma L_2(\Bbb F_q)$. Let $g_1(x)$ and $g_2(x)$ be two polynomials satisfying the condition in (\ref{eq 2.1}). If $\deg g_1(x)+\deg g_2(x)<n$, then $g(x)=g_1(x)g_2(x)$ satisfies the condition in (\ref{eq 2.1}).
\end{lem}

In Lemma \ref{lem cyclic} and \ref{lem g(x)}, a method for constructing cyclic expurgated or extended Goppa codes is described, but their generator polynomials are not determined.
Naturally,  we shall consider it.

{\bf Question:} How are their generator polynomials given in Lemmas \ref{lem cyclic} and \ref{lem g(x)}?

In detail,
   for $M=\left( \left(\begin{array}{cc}
a & b \\
c & d
\end{array}\right),\sigma^j \right) \in P\Gamma L_2(\Bbb F_q)$ and $g(x)=\sum\limits_{i=0}^r g_i x^i \in\Bbb F_q[x]$ a monic polynomial of degree $r$, there are two questions as follows:

(1) If $\sigma^j=1$, i.e., $M=A\in PGL_2(\Bbb F_q)$,  can we give generator polynomials of cyclic expurgated or extended Goppa codes?

(2) If $\sigma^j\ne 1$, i.e., $M\notin PGL_2(\Bbb F_q)$,  can we obtain generator polynomials of cyclic expurgated or extended Goppa codes?

In this paper,  we shall  determine  generator polynomials of cyclic expurgated or extended Goppa codes in the case of $\sigma^j=1$.

\section{Generator polynomials of Cyclic Expurgated or Extended Goppa Codes}

In this section, we shall give all the generator polynomials of cyclic expurgated or extended Goppa codes in the case of $\sigma^j=1$ of Section 2.3.

First of all, we shall describe the partitions of $\overline{\Bbb F}_{q^2}$ and $\overline{\Bbb F}_q$, respectively. Thus by Lemma \ref{lem g(x)}, we can obtain all Goppa polynomials $g(x)$ satisfying the condition in (\ref{eq 2.1}). Besides, we could select the support set $\mathcal L$ of a Goppa code such that its expurgated or extended code is cyclic.
 The results described in Lemmas 3.2 and 3.3 of \cite{LYa} are important.
 For completeness, we include their proofs here.

 Let  $
PGL_2(\Bbb F_q)=GL_2(\Bbb F_q)/\left\{k E_2: k \in \Bbb F_q^*\right\}
$ be projective  general linear group.  If $\widetilde A\in PGL_2(\Bbb F_q)$ and $A_1=\left(\begin{array}{ll}a_1 &b_1\\c_1&d_1\end{array}\right)$ is a representative element of the class $\widetilde A$, then $|A_1|=a_1d_1+b_1c_1=\Delta=\delta^2\ne 0$, $\delta\in\Bbb F_q$.  For convenience, we always choose the representative element $A=\frac 1\delta A_1$ of the class $\widetilde A$ with $|A|=1$.

\begin{lem}\label{lem LX1}(\cite[Lemma 3.2]{LYa})
Let
$A=\left(\begin{array}{cc}
a & b \\
c & d
\end{array}\right) \in PGL_2(\Bbb F_q)$ and $|A|=ad+bc=1$.  Then

(1) $\ord(A)=2$ if and only if $a=d$ and $(b,c)\ne (0,0)$.

(2) If $\ord(A)=n>2$  and $|\lambda E_2-A|=(\lambda-\rho)(\lambda-\rho^{-1})$, where $\rho,\rho^{-1}\in\Bbb F_{q^2}$, then $\ord(A)=\ord(\rho)$ and either $n\mid (q-1)$ or $n\mid (q+1)$.
\end{lem}
\begin{proof}
(1) For $A=\left(\begin{array}{cc}
a & b \\
c & d
\end{array}\right) \in PGL_2(\Bbb F_q)$, it is clear that $\ord(A)=2$ if and only if
$A^2=\left(\begin{array}{cc}
a^2+bc & b(a+d) \\
c(a+d) & d^2+bc
\end{array}\right)= k E_2$ and  $A\ne k'E_2$, $k, k'\in \Bbb F_q^*$,  if and only if  $a=d$ and $(b, c)\ne (0,0)$.

(2) If $\ord(A)=n>2$, then by $|A|=1$,
$$
|\lambda E_2-A| = \left|\begin{array}{cccc}
 \lambda+a &     b  \\
    c      &  \lambda+d
\end{array}\right| = \lambda^2+(a+d)\lambda+1 = (\lambda-\rho)(\lambda-\rho^{-1}),
$$
where $\rho,\rho^{-1} \in \Bbb F_{q^2}^*$ are two eigenvalues of $A$ such that $\rho+\rho^{-1}=a+d$. In fact, if the above characteristic polynomial is reducible over $\Bbb F_q$, then $\ord(\rho)=\ord(\rho^{-1})=n$ and $n\mid(q-1)$; if it is irreducible over $\Bbb F_q$, then $\ord(\rho)=\ord(\rho^{-1})=n$ and $n\mid(q+1)$.

 By (1) and $a\ne d$,  $\rho \neq \rho^{-1}$.  Then there is an invertible matrix $P\in PGL_2(\Bbb F_{q^2})$ such that $P^{-1}AP=\left(\begin{array}{cc}
\rho & 0 \\
0 & \rho^{-1}
\end{array}\right)$. It has that $A^k=P\left(\begin{array}{cc}
\rho^k & 0 \\
0 & \rho^{-k}
\end{array}\right)P^{-1}=cE_2, c \in \Bbb F_{q}^*,$ if and only if $\rho^k=\rho^{-k}$ if and only if $\rho^{2k}=1$. Hence $\ord(A) = \ord(\rho^2)=\ord(\rho)=n$.
\end{proof}

In the following, we always assume that $c \ne 0$ and $\ord(A)= n>2$, which implies $\rho \neq \rho^{-1}$. It is worth mentioning that we just consider the case $c \ne 0$ in the case of $\sigma^j=1$ of Section 2.3 because  it is  very simple if $c=0$. So  we omit it here.
\begin{lem}\label{lem LX2}(\cite[Lemma 3.3]{LYa})
Let
$A=\left(\begin{array}{cc}
a & b \\
c & d
\end{array}\right) \in PGL_2(\Bbb F_q)$ be of order $n>2$ with $|A|=ad+bc=1$ and $c\ne 0$.  Let $\alpha$ be an element of an extension of $\Bbb F_q$. Then $A(\alpha)=\alpha$ if and only if $\alpha \in \{\frac{a+\rho}c, \frac{a+\rho^{-1}}c\}$; otherwise, $n$ is the least positive integer such that $A^n(\alpha)=\alpha$, where $\rho,\rho^{-1} \in \Bbb F_{q^2}^*$ are two eigenvalues of $A$.
\end{lem}
\begin{proof}
By Lemma \ref{lem LX1}, there is an invertible matrix $P$ such that $A=P \left(\begin{array}{cc}
\rho & 0 \\
0 & \rho^{-1}
\end{array}\right) P^{-1}$ and $A^i(\alpha)=P \left(\begin{array}{cc}
\rho^i & 0 \\
0 & \rho^{-i}
\end{array}\right) P^{-1}(\alpha), 1\leq i \leq n-1$, where
\begin{equation}\label {matrix P}
P=\left(\begin{array}{cc}
a+\rho^{-1}  & a+\rho \\
c & c
\end{array}\right).\end{equation}
In fact, since $\rho,\rho^{-1} \in \Bbb F_{q^2}^*$ are two eigenvalues of $A$,
$$
\begin{aligned}
P^{-1}AP&=\left(\begin{array}{cc}
c & a+\rho \\
c & a+\rho^{-1}
\end{array}\right)
\left(\begin{array}{cc}
a & b \\
c & d
\end{array}\right)
\left(\begin{array}{cc}
a+\rho^{-1}  & a+\rho \\
c & c
\end{array}\right)\\
&=\left(\begin{array}{cc}
c\rho  & d\rho+1 \\
c\rho^{-1} & d\rho^{-1}+1
\end{array}\right)
\left(\begin{array}{cc}
a+\rho^{-1}  & a+\rho \\
c & c
\end{array}\right)\\
&=\left(\begin{array}{cc}
c(a+d)\rho  & c(\rho^2+(a+d)\rho+1) \\
c(\rho^{-2}+(a+d)\rho^{-1}+1) & c(a+d)\rho^{-1}
\end{array}\right)\\
&=\left(\begin{array}{cc}
c(a+d)\rho  & 0 \\
0 & c(a+d)\rho^{-1}
\end{array}\right)=\left(\begin{array}{cc}
\rho  & 0 \\
0 & \rho^{-1}
\end{array}\right).
\end{aligned}
$$

For $1\leq i \leq n-1, A^i(\alpha)=\alpha$ if and only if $\rho^{2i}P^{-1}(\alpha)=P^{-1}(\alpha)$. By $\ord(\rho)=n$  odd, $A^i(\alpha)=\alpha$ if and only if $P^{-1}(\alpha) \in \{0, \infty\}$ if and only if $ \alpha \in \{\frac{a+\rho}c, \frac{a+\rho^{-1}}c\}$.

Therefore, $A(\alpha)=\alpha$ if and only if $\alpha \in \{\frac{a+\rho}c, \frac{a+\rho^{-1}}c\}$; otherwise, $n$ is the least positive integer such that $A^n(\alpha)=\alpha$.
\end{proof}
\begin{cor}\label{cor orbit 1}
Let $A$ be a matrix as in Lemma \ref{lem LX2} and $G=\langle A\rangle$ a subgroup of $PGL_2(\Bbb F_q)$. For each $\alpha \in \overline{\Bbb F}_{q^2}$, let $O_{\alpha}=\{A^i(\alpha) : A^i \in G\}$ be the orbit of $\alpha$ in $\overline{\Bbb F}_{q^2}$ under the action of $G$. Then there is a partition of $\overline{\Bbb F}_{q^2}$ as follows:
$$
\overline{\Bbb F}_{q^2}=\{\frac{a+\rho}c\}\bigcup \{\frac{a+\rho^{-1}}c\}\bigcup_{i=1}^{\frac{q^2-1}n} O_{\alpha_i},
$$
where $\alpha_i \notin \{\frac{a+\rho}c, \frac{a+\rho^{-1}}c\}$ and $|O_{\alpha_i}|=\ord(A)=n$.
Moreover,
$$O_{\infty}=\{P(\rho), P(\rho^2),\ldots, P(\rho^{n-1}),P(1)=\infty\},$$ where $P$ is defined as (\ref{matrix P}).
\end{cor}
\begin{proof} Let $G=\langle A\rangle $ be a cyclic subgroup of $PGL_2(\Bbb F_{q})$ and $\overline {\Bbb F}_{q^2}$ be the projective line set, then there is an action of $G$ on $\overline {\Bbb F}_{q^2}$:
\begin{eqnarray*} G\times \overline{\Bbb F}_{q^2} &&\rightarrow \overline{\Bbb F}_{q^2},\\
(A^i,\alpha) &&\mapsto A^i(\alpha).
\end{eqnarray*}
Then $\overline{\Bbb F}_{q^2}$ has a partition (in details, see \cite[Proposition 272]{RE}).
By Lemma \ref{lem LX2}, $\frac{a+\rho}c$ and $\frac{a+\rho^{-1}}c$ are the only two fixed points under the action of group $G=\langle A\rangle$; since $n$ is the least positive integer such that $A^n(\alpha)=\alpha$ if $\alpha \notin \{\frac{a+\rho}c, \frac{a+\rho^{-1}}c\}$,  $O_{\alpha}=G(\alpha)=\{\alpha,A(\alpha),\ldots,A^{n-1}(\alpha)\}$. Hence there is a partition of $\overline{\Bbb F}_{q^2}$ as follows:
$$
\overline{\Bbb F}_{q^2}=\{\frac{a+\rho}c\}\bigcup \{\frac{a+\rho^{-1}}c\}\bigcup_{i=1}^{\frac{q^2-1}n} O_{\alpha_i},
$$
where $\alpha_i \notin \{\frac{a+\rho}c, \frac{a+\rho^{-1}}c\}$ and $|O_{\alpha_i}|=\ord(A)=n$.

Let $\Lambda=\left(\begin{array}{cc}\rho &0\\0 &\rho^{-1}\end{array}\right)$ and $G'=\langle \Lambda\rangle$ a subgroup of $PGL_2(\Bbb F_{q^2})$, then $G'$ acts on the set $\overline {\Bbb F}_{q^2}$, so for $\beta\notin \{0,\infty\}$, there is an orbit of $\beta$ under the action of $G'$:
$$O_\beta=\{\beta, \rho\beta, \ldots, \rho^{n-1}\beta\}.$$
By $A=P\Lambda P^{-1}$, $P(0)=\frac{a+\rho}c$, $P(\infty)=\frac{a+\rho^{-1}}c$. Hence   for $\beta\notin\{0, \infty\}$, there is an orbit of $P(\beta) $ under the action of $G=\langle A\rangle $:
$$O_{P(\beta)}=\{P(\beta), P(\rho \beta),\ldots, P(\rho^{n-1} \beta)\}, $$
so $$O_{\infty}=\{ P(\rho ),\ldots, P(\rho^{n-1} ),P(1)=\infty\}. $$
\end{proof}

\begin{cor}\label{cor orbit 2}
Let $A$ be a matrix as in Lemma \ref{lem LX2} and $G=\langle A\rangle$ a subgroup of $PGL_2(\Bbb F_q)$. For each $\alpha \in \overline{\Bbb F}_{q}$, let $O_{\alpha}=\{A^i(\alpha) : A^i \in G\}$ be the orbit of $\alpha$ in $\overline{\Bbb F}_{q}$ under the action of $G$.

(1) If $\lambda^2+(a+d)\lambda+1 = (\lambda-\rho)(\lambda-\rho^{-1})$ is reducible over $\Bbb F_q$, then there is a partition of $\overline{\Bbb F}_{q}$ as follows:
$$
\overline{\Bbb F}_{q}=\{\frac{a+\rho}c\}\bigcup \{\frac{a+\rho^{-1}}c\}\bigcup_{i=1}^{\frac{q-1}n} O_{\alpha_i},
$$
where $\alpha_i \notin \{\frac{a+\rho}c, \frac{a+\rho^{-1}}c\}$ and $|O_{\alpha_i}|=\ord(A)=n$.

(2) If $\lambda^2+(a+d)\lambda+1 = (\lambda-\rho)(\lambda-\rho^{-1})$ is irreducible over $\Bbb F_q$, then there is a partition of $\overline{\Bbb F}_{q}$ as follows:
$$
\overline{\Bbb F}_{q}=\bigcup_{i=1}^{\frac{q+1}n} O_{\alpha_i},
$$
where $|O_{\alpha_i}|=\ord(A)=n$.
\end{cor}

\begin{prop} \label{prop g}
Let $A$ be a matrix as in Lemma \ref{lem LX2}, $\lambda^2+(a+d)\lambda+1=(\lambda-\rho)(\lambda-\rho^{-1})$,  $\rho,\rho^{-1}\in \Bbb F_{q^2}$, $G=\langle A\rangle$ a subgroup of $PGL_2(\Bbb F_q)$, and $g(x)=\sum\limits_{i=0}^r g_i x^i \in \Bbb F_{q^2}[x]$  a monic polynomial satisfying the condition in (\ref{eq 2.1}). Then both the expurgated Goppa code $\widetilde{\Gamma}(\mathcal L, g)$ and the extended Goppa code $\overline{\Gamma}(\mathcal L', g)$ are cyclic of length $n$, where
$$
\mathcal L'=O_{\infty}\backslash \{\infty\}=\{ P(\rho ),\ldots, P(\rho^{n-1} )\}=\{\gamma, A(\gamma), \ldots, A^{n-2}(\gamma)\} \subseteq \Bbb F_{q^2},
$$
 and
 $$\mathcal L=O_{\alpha}=\{\alpha_1=\alpha, \alpha_2=A(\alpha),\ldots,\alpha_n=A^{n-1}(\alpha)\} \subseteq \Bbb F_{q^2}$$
is the orbit of $\alpha$ under the action of $G$ if $\alpha \notin \{\frac{a+\rho}c,\frac{a+\rho^{-1}}c\}\cup O_{\infty}$.

(1) If the degree $r$ of $g(x)$ is not less than  $n-1$, i.e., $r\ge n-1$,  then both the expurgated Goppa code $\widetilde{\Gamma}(\mathcal L, g)$ and the extended Goppa code $\overline{\Gamma}(\mathcal L', g)$ are zero.

(2) If the degree $r$ of $g(x)$ is less than $n-1$, then
\begin{equation} \label{eq g}
g(x)= g_1(x)^sg_2(x)^t,
\end{equation}
where
\begin{equation} \label{eq g_i}
g_1(x)=x+\frac{a+\rho}c,~ g_2(x)=x+\frac{a+\rho^{-1}}c,
\end{equation}
and  $s+t <n-1$.
\end{prop}
\begin{proof}
(1) If $r\ge n-1$  and  $\widetilde H$ is defined as (\ref{widetilde H}),
 then $\mbox{rank}_{\Bbb F_{q^2}}(\widetilde H)=n$ and
\begin{eqnarray*}&&\widetilde{\Gamma}(\mathcal L, g)=\left\{\boldsymbol c=(c_1, \ldots, c_n) \in \Bbb F_2^n: \widetilde H \boldsymbol c^\mathrm T=0\right\}\\&& \subseteq
\left\{\boldsymbol c=(c_1, \ldots, c_n) \in \Bbb F_{q^2}^n: \widetilde H \boldsymbol c^\mathrm T=0\right\}=\{\boldsymbol 0\},\end{eqnarray*}
Hence the expurgated Goppa code $\widetilde{\Gamma}(\mathcal L, g)$ is zero.
 Similarly, we can prove that the extended Goppa code $\overline{\Gamma}(\mathcal L', g)$ is also  zero.

(2) If the degree $r$ of  $g(x)$ is less than  $n-1$ and $g(x)$ satisfies the condition in (\ref{eq 2.1}),
  then by  Lemma \ref{lem g(x)}, there is not any root of $g(x)$ in an orbit $O_{\alpha_i}$ with $|O_{\alpha_i}|=n$ in Corollary \ref{cor orbit 1}. Hence $g(x)=g_1(x)^sg_2(x)^t \in \Bbb F_{q^2}[x]$, where $ g_1(x)=x+\frac{a+\rho}c,~ g_2(x)=x+\frac{a+\rho^{-1}}c$, $s,t \in \mathbb{N}$,  and $s+t <n-1$.
\end{proof}

\begin{rem} \label{rem g}
In Proposition \ref{prop g}, if $g(x)=\sum\limits_{i=0}^r g_i x^i \in \Bbb F_q[x]$ is a monic polynomial satisfying the condition in (\ref{eq 2.1}), then
\begin{equation}\label{remark g irr}
g(x)=\begin{cases} g_1(x)^sg_2(x)^t, &\text{if}~ \lambda^2+(a+d)\lambda+1 ~\text{is reducible over}~ \Bbb F_q,\\
(g_1(x)g_2(x))^s, &\text{if}~ \lambda^2+(a+d)\lambda+1 ~\text{is irreducible over} ~\Bbb F_q,\end{cases}
\end{equation}
where $s,t \in \mathbb{N}$, $s+t <n-1$ if $\lambda^2+(a+d)\lambda+1$ is reducible over $\Bbb F_q$, and $2s<n-1$ if $\lambda^2+(a+d)\lambda+1$ is irreducible over $\Bbb F_q$.

Indeed, if $\lambda^2+(a+d)\lambda+1$ is reducible over $\Bbb F_q$, then $g_1(x)=x+\frac{a+\rho}c$, $g_2(x)=x+\frac{a+\rho^{-1}}c \in \Bbb F_q[x]$, and $g(x)=g_1(x)^sg_2(x)^t \in \Bbb F_q[x]$. If $\lambda^2+(a+d)\lambda+1$ is irreducible over $\Bbb F_q$, then $\ord(\rho)=\ord(A)=n\mid (q+1)$, $(\frac {a+\rho}c)^q=\frac {a+\rho^{-1}}c$, $g_1(x)^sg_2(x)^t=g(x)=(g(x))^q=g_2(x)^sg_1(x)^t$ by $g(x)\in \Bbb F_q[x]$, thus $s=t$.
\end{rem}

 In the following, we always assume that the degree $r$ of $g(x)$ is less than $n-1$. We shall find generator polynomials of cyclic expurgated or extended Goppa codes with Goppa polynomial $g(x)= g_1(x)^sg_2(x)^t$, $s,t \in \mathbb{N}$ and $s+t <n-1$. Specifically, we shall consider the Goppa polynomial $g(x)$ in  three cases:

 (1) $g(x)=g_1(x)$ or $g_2(x)$;

 (2) $g(x)=g_1(x)^s$ or $g_2(x)^t$, $s,t \in \mathbb{N}$ and $1\leq s, t<n-1$;

 (3) $g(x)=g_1(x)^sg_2(x)^t$, $s,t \in \mathbb{N}$ and $s+ t<n-1$;\\
  where $g_1(x)$ and $g_2(x)$ are just defined as (\ref{eq g_i}).
\subsection{$g(x)=g_1(x)$ or $g_2(x)$.}\

In the subsection, we shall assume that $g(x)=g_1(x)$ or $g_2(x)$ and find generator polynomials of cyclic expurgated or extended Goppa codes in the case of $\sigma^j=1$ of Section 2.3.

\begin{thm} \label{Th 1}
Let $A=\left(\begin{array}{cc}
a & b \\
c & d
\end{array}\right) \in PGL_2(\Bbb F_q)$ be of order $n>2$ with $|A|=ad+bc=1$ and $c\ne 0$,  $G=\langle A\rangle$ a subgroup of $PGL_2(\Bbb F_q)$, and  $|\lambda E_2-A|=\lambda^2+(a+d)\lambda+1=(\lambda-\rho)(\lambda-\rho^{-1})$  reducible over $\Bbb F_q$.
Let
$$
\mathcal L'=O_{\infty}\backslash \{\infty\}=\{ P(\rho ),\ldots, P(\rho^{n-1} )\}=\{\gamma, A(\gamma), \ldots, A^{n-2}(\gamma)\} \subseteq \Bbb F_q
$$
 and
 $$\mathcal L=O_{\alpha}=\{\alpha_1=\alpha, \alpha_2=A(\alpha),\ldots,\alpha_n=A^{n-1}(\alpha)\} \subseteq \Bbb F_{q}$$
 be the orbit of $\alpha$ under the action of $G$ if $\alpha \notin \{\frac{a+\rho}c,\frac{a+\rho^{-1}}c\}\cup O_{\infty}$.

(1) If $g_1(x)=x+\frac{a+\rho}c \in \Bbb F_q[x]$, then both the expurgated Goppa code $\widetilde{\Gamma}(\mathcal L, g_1)$ and the extended Goppa code $\overline{\Gamma}(\mathcal L', g_1)$ are cyclic of length $n$ and  their generator polynomials are  $u_1(x)=(x+1)m_{\rho^{-1}}(x)$, where $m_{\rho^{-1}}(x)$ is the minimal polynomial of $\rho^{-1}$ over $\Bbb F_2$.

(2) If $g_2(x)=x+\frac{a+\rho^{-1}}c \in \Bbb F_q[x]$, then both the expurgated Goppa code $\widetilde{\Gamma}(\mathcal L, g_2)$ and the extended Goppa code $\overline{\Gamma}(\mathcal L', g_2)$  are cyclic of length $n$ and  their generator polynomials are  $u_2(x)=(x+1)m_\rho(x)$, where $m_\rho(x)$ is the minimal polynomial of $\rho$ over $\Bbb F_2$.
\end{thm}
\begin{proof}
 Let $\lambda^2+(a+d)\lambda+1=(\lambda-\rho)(\lambda-\rho^{-1})$ be reducible over $\Bbb F_q$, then by Proposition \ref{prop g},
$$
g_1(x)=x+\frac{a+\rho}c,~~ g_2(x)=x+\frac{a+\rho^{-1}}c \in \Bbb F_q[x]
$$
satisfy the condition in  (\ref{eq 2.1}). Hence by Lemma \ref{lem cyclic}, the expurgated Goppa codes $\widetilde{\Gamma}(\mathcal L, g_1)$, $\widetilde{\Gamma}(\mathcal L, g_2)$ and the extended Goppa codes $\overline{\Gamma}(\mathcal L', g_1)$, $\overline{\Gamma}(\mathcal L', g_2)$ are all cyclic of length $n$.
 In the following, we only need to find the generator polynomial of $\widetilde{\Gamma}(\mathcal L, g_1)$. Similarly, we can also obtain those of $\overline{\Gamma}(\mathcal L', g_1)$, $\widetilde{\Gamma}(\mathcal L, g_2)$ and $\overline{\Gamma}(\mathcal L', g_2)$.

For the cyclic expurgated Goppa code $\widetilde{\Gamma}(\mathcal L, g_1)$, we shall prove that the polynomial $(x+1)m_{\rho^{-1}}(x)$ is its generator polynomial, where $m_{\rho^{-1}}(x)$ is the minimal polynomial of $\rho^{-1}$ over $\Bbb F_2$. Now we consider the parity-check matrix of $\widetilde{\Gamma}(\mathcal L, g_1)$ as follows:
$$
\widetilde H=\left(\begin{array}{cccc}
\frac{1}{\alpha_1+\frac{a+\rho}c} &  \cdots & \frac{1}{\alpha_n+\frac{a+\rho}c} \\
1 &  \cdots & 1
\end{array}\right).
$$
Hence $\boldsymbol c=(c_1,\ldots,c_n) \in \widetilde{\Gamma}(\mathcal L, g_1)$ if and only if $\widetilde H \boldsymbol c^\mathrm T=0$ if and only if $\boldsymbol c$ is a solution of the system of equations
\begin{equation} \label{Th 1,eq 1}
\left\{\begin{array}{ll}
&\frac{c_1}{\alpha_1+\frac{a+\rho}c}+\cdots+\frac{c_n}{\alpha_n+\frac{a+\rho}c}=0, \\
&c_1+c_2+\cdots+c_n=0.
\end{array}\right.
\end{equation}

In the following, we shall prove  that $\boldsymbol c=(c_1,\ldots,c_n) \in \widetilde{\Gamma}(\mathcal L, g_1)$ if and only if $\boldsymbol c$ is a solution of the system of equations
\begin{equation} \label{Th1,eq 2}
\left\{\begin{array}{ll}
&c_1+c_2e+\cdots+c_ne^{n-1}=0, \\
&c_1+c_2+\cdots+c_n=0, \end{array}\right.
\end{equation}
 where $e=\rho^{-2}$. It means that the system in (\ref{Th 1,eq 1}) is equivalent to the system in (\ref{Th1,eq 2}).

In (\ref{Th 1,eq 1}), $\frac 1{\alpha_i+\frac{a+\rho}c}=B(\alpha_i)$, $1\le i\le  n$, where
 $B=\left(\begin{array}{cc}
0 & 1 \\
1 & \frac{a+\rho}c
\end{array}\right)$.  On the other hand,
$$
D=BAB^{-1}=\left(\begin{array}{cc}
a+d+\rho & c \\
\frac{\rho^2+(a+d)\rho+1}c & \rho
\end{array}\right)=\left(\begin{array}{cc}
\rho^{-2} & c\rho^{-1} \\
0 & 1
\end{array}\right)\triangleq \left(\begin{array}{cc}
e & f \\
0 & 1
\end{array}\right),
$$
where $e=\rho^{-2}$ and $f=c\rho^{-1}$. Then $\ord(D)=\ord(A)=\ord(e)=n$ and
$$
D^i=BA^iB^{-1}=\left(\begin{array}{cc}
e^i & f(e^{i-1}+e^{i-2}+\cdots+1) \\
0 & 1
\end{array}\right), 1\leq i \leq n-1.
$$

By $\alpha_{i+1}=A(\alpha_i), 1\leq i \leq n-1$, the system (\ref{Th 1,eq 1}) is  equivalent to the system
$$\left\{\begin{array}{ll}
&c_1B(\alpha_1)+c_2B(A(\alpha_1))+\cdots+c_nB(A^{n-1}(\alpha_1))=0,\\
&c_1+c_2+\cdots+c_n=0.\end{array}\right.
$$
Note that $B(A^i(\alpha_1))=BA^iB^{-1}B(\alpha_1)=D^iB(\alpha_1)$, $1\le i\le n-1$.
Hence the above system is equivalent to the system
$$\left\{\begin{array}{ll}
&c_1B(\alpha_1)+c_2D(B(\alpha_1))+\cdots+c_nD^{n-1}(B(\alpha_1))=0, \\
&c_1+c_2+\cdots+c_n=0.\end{array}\right.
$$
Now we consider the first equation in the above system.
$$
\begin{aligned}
&c_1B(\alpha_1)+c_2D(B(\alpha_1))+\cdots+c_nD^{n-1}(B(\alpha_1))\\
&=(c_1+c_2e+c_3e^2+\cdots+c_ne^{n-1})B(\alpha_1)\\ &+ \frac f{e-1}(c_2(e-1)+c_3(e^2-1)+\cdots+c_n(e^{n-1}-1))\\
&= (c_1+c_2e+c_3e^2+\cdots+c_ne^{n-1})(B(\alpha_1)+\frac f{e-1})=0.
\end{aligned}
$$

  Suppose that $B(\alpha_1)+ \frac f{e-1}=0$, i.e., $\frac 1{\alpha_1+\frac{a+\rho}c} = \frac{c\rho^{-1}}{\rho^{-2}-1}$.  Then   $\alpha_1=\frac{a+\rho^{-1}}c$, which is contradictory with the  choice of $\alpha_1$. Therefore,  the  system in (\ref{Th 1,eq 1}) is equivalent to the system in (\ref{Th1,eq 2}).

  Moreover, let $m_{\rho^{-1}}(x)$ be the minimal polynomial of $\rho^{-1}$ over $\Bbb F_2$, then it is also the minimal polynomial of  $e=\rho^{-2}$ over $\Bbb F_2$. Since $\widetilde{\Gamma}(\mathcal L, g_1)$ is a cyclic code of length $n$, whose elements satisfy the system (\ref{Th1,eq 2}), its generator polynomial is $u_1(x)=(x-1)m_{\rho^{-1}}(x)$.

  Similarly, we can get generator polynomial of $\overline{\Gamma}(\mathcal L', g_1)$ is $u_1(x)=(x-1)m_{\rho^{-1}}(x)$ and generator polynomials of both $\widetilde{\Gamma}(\mathcal L, g_2)$ and $\overline{\Gamma}(\mathcal L', g_2)$ are $u_2(x)=(x-1)m_{\rho}(x)$, where $m_{\rho}(x)$ is the minimal polynomial of $\rho$ over $\Bbb F_2$.
\end{proof}
\begin{rem} \label{WY}
In  \cite{WY}, $m_\rho(x)=m_{\rho^{-1}}(x)$  is equivalent to the following conditions:

(1) The polynomial $m_\rho(x)$ is self-reciprocal.

(2) There exists a positive integer $\omega$ such that $2^\omega \equiv -1 \pmod n.$

(3) Suppose $n=p_1^{l_1}p_2^{l_2}\cdots p_t^{l_t}$ is the prime factorization, where $p_1,\ldots,p_t$ are distinct odd primes and $l_1,\ldots,l_t$ are positive integers. By this time, $v_2(d_1)=\cdots=v_2(d_t)=\delta$, where $d_i$ is the order of 2 modulo $p_i$, $v_2(d_i)$ denotes the $2$-adic valuation of $d_i$, $1\leq i\leq t$, and $\delta$ is a positive integer.
\end{rem}

\begin{rem}\label{conclusion 1} In Theorem \ref{Th 1},
if   $A=\left(\begin{array}{cc}
a & b \\
c & d
\end{array}\right)\in PGL_2(\Bbb F_q)$ and $\lambda^2+(a+d)\lambda+1=(\lambda-\rho)(\lambda-\rho^{-1})$ is reducible over $\Bbb F_q$, then the polynomial $g(x)$ of degree $1$ that satisfies the condition in (\ref{eq 2.1}) can only be $g_1(x)=x+\frac{a+\rho}c$ or $g_2(x)=x+\frac{a+\rho^{-1}}c$.
Hence  all possible cyclic expurgated Goppa codes with Goppa polynomial of degree $1$ are $\widetilde{\Gamma}(\mathcal L, g_1)$ and $\widetilde{\Gamma}(\mathcal L, g_2)$; all possible cyclic extended Goppa codes with Goppa polynomial of degree $1$ are $\overline{\Gamma}(\mathcal L', g_1)$ and $\overline{\Gamma}(\mathcal L', g_2)$.
\end{rem}
\begin{cor}\label{Th1,cor,irr}
Let $A\in PGL_2(\Bbb F_q)$ be the matrix, $G=\langle A\rangle$ the subgroup  in Theorem \ref{Th 1}, and  $|\lambda E_2-A|=\lambda^2+(a+d)\lambda+1=(\lambda-\rho)(\lambda-\rho^{-1})$ irreducible over $\Bbb F_q$. Let
$$
\mathcal L'=O_{\infty}\backslash \{\infty\}=\{ P(\rho ),\ldots, P(\rho^{n-1} )\}=\{\gamma, A(\gamma), \ldots, A^{n-2}(\gamma)\} \subseteq \Bbb F_{q^2}
$$
and
$$\mathcal L=O_{\alpha}=\{\alpha_1=\alpha, \alpha_2=A(\alpha),\ldots,\alpha_n=A^{n-1}(\alpha)\} \subseteq \Bbb F_{q^2}$$
 be the orbit of $\alpha$ under the action of $G$ if $\alpha \notin \{\frac{a+\rho}c,\frac{a+\rho^{-1}}c\}\cup O_{\infty}$.

 Then the conclusions of Theorem \ref{Th 1} still hold but with the Goppa polynomials $g_1(x), g_2(x)\in \Bbb F_{q^2}[x]$. Moreover, let $\lambda^2+(a+d)\lambda+1$ be irreducible over $\Bbb F_q$, $\ord(A)=n\mid (q+1)$, then by Remark \ref{WY}, $m_\rho(x)=m_{\rho^{-1}}(x)$.
\end{cor}

By Lemma \ref{lem Goppa}, we can get the following result.
\begin{cor}\label{Th1,cor,square}
In Theorem \ref{Th 1}, $\widetilde {\Gamma}(\mathcal L, g_i^2)=\widetilde {\Gamma}(\mathcal L, g_i)=\langle h_i(x)\rangle=\overline {\Gamma}(\mathcal L', g_i^2)=\overline {\Gamma}(\mathcal L', g_i)$, $i=1,2$,  are all  cyclic codes of length $n$.
\end{cor}

\begin{exa} \label{eg 1} 
Let $\Bbb F_{2^6}^*=\langle \alpha \rangle$ and $A=\left(\begin{array}{cc}
\alpha^5 & \alpha^{43} \\
\alpha^{13} & \alpha^{59}
\end{array}\right) \in PGL_2(\Bbb F_{2^6})$, then $|A|=1$, $\ord(A)=21>2$, and $\lambda^2+(\alpha^5+\alpha^{59})\lambda+1=(\lambda-\rho)(\lambda-\rho^{-1}) \in \Bbb F_{2^6}[x]$, where $\rho=\alpha^3, \rho^{-1}=\alpha^{-3}=\alpha^{60} \in \Bbb F_{2^6}^*$ are two eigenvalues of $A$. The minimal polynomials of $\rho^{-1}$ and $\rho$ over $\Bbb F_2$ are $m_{\rho^{-1}}(x)=x^6+x^4+x^2+x+1$ and $m_\rho(x)=x^6+x^5+x^4+x^2+1$, respectively. Let
$$
\begin{aligned}
\mathcal L'=O_{\infty}\backslash \{\infty\}=
&\{\alpha^{55},\alpha^{59},\alpha^{23},\alpha^{15},\alpha^{7},\alpha^{34},\alpha^{38},
0,\alpha^{47},\alpha^{18},\alpha^{17},
\alpha^{53},\alpha^{10},\alpha^{42},\\
&\alpha^{51},\alpha^{20},\alpha^{40},\alpha^{13},\alpha^{12},\alpha^{46}\} \subseteq \Bbb F_{2^6}
\end{aligned}
$$
and
$$
\begin{aligned}
\mathcal L=O_{\alpha}=
&\{\alpha,\alpha^{62},\alpha^{11},\alpha^{60},\alpha^{14},\alpha^{5},\alpha^{6},
\alpha^{32},\alpha^{2},\alpha^{8},
\alpha^{27},\alpha^{30},\alpha^{48},\\ &\alpha^{58},\alpha^{52},\alpha^{37},\alpha^{21},\alpha^{36},\alpha^{44},\alpha^{26},
\alpha^{50}\} \subseteq \Bbb F_{2^6},
\end{aligned}
$$
where $|\mathcal L'|=20$ and $|\mathcal L|=21$ = $\ord (A)$.

(1) If $g_1(x)=x+\frac{\alpha^5+\alpha^3}{\alpha^{13}}$ and $g_1(x)^2=(x+\frac{\alpha^5+\alpha^3}{\alpha^{13}})^2$, we can obtain that the expurgated Goppa codes $\widetilde{\Gamma}(\mathcal L, g_1)$ and $\widetilde{\Gamma}(\mathcal L, g_1^2)$, the extended Goppa codes $\overline{\Gamma}(\mathcal L', g_1)$ and $\overline{\Gamma}(\mathcal L', g_1^2)$ are all $[21,14,4]$ cyclic codes and their generator polynomials are $(x+1)(x^6+x^4+x^2+x+1)=(x+1)m_{\rho^{-1}}(x)$ by the Magma program.

(2) If $g_2(x)=x+\frac{\alpha^5+\alpha^{-3}}{\alpha^{13}}$ and $g_2(x)^2=(x+\frac{\alpha^5+\alpha^{-3}}{\alpha^{13}})^2$, we can obtain that the expurgated Goppa codes $\widetilde{\Gamma}(\mathcal L, g_2)$ and $\widetilde{\Gamma}(\mathcal L, g_2^2)$, the extended Goppa codes $\overline{\Gamma}(\mathcal L', g_2)$ and $\overline{\Gamma}(\mathcal L', g_2^2)$ are all $[21,14,4]$ cyclic codes and their generator polynomials are $(x+1)(x^6+x^5+x^4+x^2+1)=(x+1)m_\rho(x)$ by the Magma program.
\end{exa}

\begin{exa} 
Let $\Bbb F_{2^6}^*=\langle \alpha \rangle$ and $A=\left(\begin{array}{cc}
\alpha^7 & 0 \\
1 & \alpha^{-7}
\end{array}\right) \in PGL_2(\Bbb F_{2^6})$, then $|A|=1$, $\ord(A)=9>2$, and $\lambda^2+(\alpha^7+\alpha^{-7})\lambda+1=(\lambda-\rho)(\lambda-\rho^{-1}) \in \Bbb F_{2^6}[x]$, where $\rho=\alpha^7, \rho^{-1}=\alpha^{-7} \in \Bbb F_{2^6}^*$ are two eigenvalues of $A$. Both the minimal polynomial of $\rho$ and $\rho^{-1}$ over $\Bbb F_2$ are $m_\rho(x)=m_{\rho^{-1}}(x)=x^6+x^3+1$. Let
$$
\begin{aligned}
\mathcal L'=O_{\infty}\backslash \{\infty\}
=\{\alpha^7,\alpha^5,\alpha^{30},\alpha,\alpha^8,\alpha^{51},\alpha^{40},\alpha^{56}\} \subseteq \Bbb F_{2^6}
\end{aligned}
$$
and
$$
\begin{aligned}
\mathcal L
=O_{\alpha^2}=\{\alpha^2,\alpha^{52},\alpha^{35},\alpha^{28},\alpha^{38},\alpha^{16},
\alpha^{19},\alpha^{27},\alpha^{26}\}
\subseteq \Bbb F_{2^6},
\end{aligned}
$$
where $|\mathcal L'|=8$ and $|\mathcal L|=9$ = $\ord (A)$.

(1) If $g_1(x)=x$ and $g_1^2(x)=x^2$, we can obtain that the expurgated Goppa codes $\widetilde{\Gamma}(\mathcal L, g_1)$ and $\widetilde{\Gamma}(\mathcal L, g_1^2)$, the extended Goppa codes $\overline{\Gamma}(\mathcal L', g_1)$ and $\overline{\Gamma}(\mathcal L', g_1^2)$ are all $[9,2,6]$ cyclic codes and their generator polynomials are $(x+1)(x^6+x^3+1)=(x+1)m_{\rho^{-1}}(x)$ by the Magma program.

(2) If $g_2(x)=x+\alpha^7+\alpha^{-7}$ and $g_2^2(x)=(x+\alpha^7+\alpha^{-7})^2$, we can obtain that the expurgated Goppa codes $\widetilde{\Gamma}(\mathcal L, g_2)$ and $\widetilde{\Gamma}(\mathcal L, g_2^2)$, the extended Goppa codes $\overline{\Gamma}(\mathcal L', g_2)$ and $\overline{\Gamma}(\mathcal L', g_2^2)$ are all $[9,2,6]$ cyclic codes and their generator polynomials are  $(x+1)(x^6+x^3+1)=(x+1)m_\rho(x)$ by the Magma program.
\end{exa}
\begin{exa}
Let $\Bbb F_{2^4}^*=\langle \alpha \rangle$ and $\Bbb F_{2^8}^*=\langle \gamma \rangle$, where $\alpha=\gamma^{17}$. Let
$A=\left(\begin{array}{cc}
\alpha^{11} & \alpha^5 \\
\alpha^3 & \alpha^6
\end{array}\right) \in PGL_2(\Bbb F_{2^4})$, then $|A|=1$, $\ord(A)=17>2$, and $\lambda^2+(\alpha^{11}+\alpha^6)\lambda+1=(\lambda-\rho)(\lambda-\rho^{-1}) \in \Bbb F_{2^8}[x]$, where $\rho=\gamma^{45}, \rho^{-1}=\gamma^{-45}=\gamma^{210} \in \Bbb F_{2^8}^*$ are two eigenvalues of $A$. Both the minimal polynomial of $\rho$ and $\rho^{-1}$ over $\Bbb F_2$ are $m_\rho(x)=m_{\rho^{-1}}(x)=x^8+x^5+x^4+x^3+1$. Let
$$
\begin{aligned}
\mathcal L'=O_{\infty}\backslash \{\infty\}
=\{\gamma^{136},\gamma^{119},\gamma^{68},1,\gamma^{17},\gamma^{34},\gamma^{221},\gamma^{170}, \gamma^{153},\\ 0, \gamma^{238},
\gamma^{204},\gamma^{102},\gamma^{187},\gamma^{85},\gamma^{51}\} \subseteq \Bbb F_{2^8}
\end{aligned}
$$
and
$$
\begin{aligned}
\mathcal L
=O_{\gamma^3}=\{\gamma^3,\gamma^{26},\gamma^{147},\gamma^{172},\gamma^{32},
\gamma^{87},\gamma^{232},\gamma^{128}, \gamma^{241},\\ \gamma^{61}, \gamma^{144},\gamma^{191},\gamma^{39},\gamma^{175},\gamma^{38},\gamma^{25},\gamma^{78}\} \subseteq \Bbb F_{2^8}
\end{aligned}
$$
where $|\mathcal L'|=16$ and $|\mathcal L|=17=\ord (A)$.

(1) If $g_1(x)=x+\frac{\alpha^{11}+\gamma^{45}}{\alpha^3}$ and $g_1(x)^2=(x+\frac{\alpha^{11}+\gamma^{45}}{\alpha^3})^2$, we can obtain that the expurgated Goppa codes $\widetilde{\Gamma}(\mathcal L, g_1)$ and $\widetilde{\Gamma}(\mathcal L, g_1^2)$, the extended Goppa codes $\overline{\Gamma}(\mathcal L', g_1)$ and $\overline{\Gamma}(\mathcal L', g_1^2)$ are all $[17,8,6]$ cyclic codes and their generator polynomials are $(x+1)(x^8+x^5+x^4+x^3+1)=(x+1)m_{\rho^{-1}}(x)$ by the Magma program.

(2) If $g_2(x)=x+\frac{\alpha^{11}+\gamma^{210}}{\alpha^3}$ and $g_2(x)^2=(x+\frac{\alpha^{11}+\gamma^{210}}{\alpha^3})^2$, we can obtain that the expurgated Goppa codes $\widetilde{\Gamma}(\mathcal L, g_2)$ and $\widetilde{\Gamma}(\mathcal L, g_2^2)$, the extended Goppa codes $\overline{\Gamma}(\mathcal L', g_2)$ and $\overline{\Gamma}(\mathcal L', g_2^2)$ are all $[17,8,6]$ cyclic codes and their generator polynomials are $(x+1)(x^8+x^5+x^4+x^3+1)=(x+1)m_\rho(x)$ by the Magma program.
\end{exa}

\subsection {$g(x)=g_1(x)^s$ \text{or} $g_2(x)^t$, $s,t \in \mathbb{N}$ \text{and} $s, t<n-1$}\

In the subsection, we shall assume that $g(x)=g_1(x)^s$ or $g_2(x)^t$, where $s,t \in \mathbb{N}$ and $s, t<n-1$.
 We shall find generator polynomials of cyclic expurgated or extended Goppa codes in the case of $\sigma^j=1$ of Section 2.3.
There is a surprising result that the expurgated or extended Goppa codes  are all BCH codes, which are an important class of cyclic codes.
\begin{defn}\label{BCH} (\cite[Definition 9.44]{LN})
Let $a$ be a nonnegative integer and  $\alpha\in \Bbb F_{q^m}$ be a primitive $n$-th root of unity. A BCH code $\mathcal C=\langle u(x)\rangle$ over $\Bbb F_q$ of length $n$ and designed distance $\delta$, $2\le \delta \le n$, is a cyclic code with the generator polynomial:
$$u(x)=\mbox{lcm}(m^{(a)}(x),m^{(a+1)}(x),\ldots, m^{(a+\delta-2)}(x)),$$
where each $m^{(i)}(x)$ is  the minimal polynomial of $\alpha^i$ over $\Bbb F_q$. Then the minimum distance of the BCH code $\mathcal C$ is at least $\delta$.
\end{defn}

\begin{thm} \label{Th 2}
Let $A=\left(\begin{array}{cc}
a & b \\
c & d
\end{array}\right) \in PGL_2(\Bbb F_q)$ be of order $n>2$ with $|A|=ad+bc=1$ and $c\ne 0$,  $G=\langle A\rangle$ a subgroup of $PGL_2(\Bbb F_q)$, and $|\lambda E_2-A|=\lambda^2+(a+d)\lambda+1=(\lambda-\rho)(\lambda-\rho^{-1})$  reducible over $\Bbb F_q$.
Let
$$
\mathcal L'=O_{\infty}\backslash \{\infty\}=\{ P(\rho),\ldots, P(\rho^{n-1} )\}=\{\gamma, A(\gamma), \ldots, A^{n-2}(\gamma)\} \subseteq \Bbb F_q
$$
and
$$\mathcal L=O_{\alpha}=\{\alpha_1=\alpha, \alpha_2=A(\alpha),\ldots,\alpha_n=A^{n-1}(\alpha)\} \subseteq \Bbb F_{q}$$
be the orbit of $\alpha$ under the action of $G$ if $\alpha \notin \{\frac{a+\rho}c,\frac{a+\rho^{-1}}c\}\cup O_{\infty}$.

(1) If $g_1(x)^s=(x+\frac{a+\rho}c)^s \in \Bbb F_q[x]$, then both the expurgated Goppa code $\widetilde{\Gamma}(\mathcal L, g_1^s)$ and the extended Goppa code $\overline{\Gamma}(\mathcal L', g_1^s)$ are  BCH codes of length $n$ with designed distances $\delta=s+2$ and  the generator polynomials:
\begin{equation} \label{eq u_1}
u_1(x)=(x+1)\mbox{lcm}\{m_{\rho^{-i}}(x): i=1,2,\ldots, s\},
\end{equation}
where each $m_{\rho^{-i}}(x)$ is  the minimal polynomial of $\rho^{-i}$ over $\Bbb F_2$.

(2) If $g_2(x)^t=(x+\frac{a+\rho^{-1}}c)^t \in \Bbb F_q[x]$, then both the expurgated Goppa code $\widetilde{\Gamma}(\mathcal L, g_2^t)$ and the extended Goppa code $\overline{\Gamma}(\mathcal L', g_2^t)$  are BCH codes of length $n$ with designed distances $\delta=t+2$ and  the  generator polynomials:
\begin{equation} \label{eq u_2}
u_2(x)=(x+1)\mbox{lcm}\{m_{\rho^i}(x): i=1,2,\ldots, t\},
\end{equation}
where each $m_{\rho^i}(x)$ is  the minimal polynomial of $\rho^i$ over $\Bbb F_2$.

Moreover, if $\deg u_1(x)=n$, then both $\widetilde{\Gamma}(\mathcal L, g_1^s)$ and $\overline{\Gamma}(\mathcal L', g_1^s)$ are zero; similarly, both $\widetilde{\Gamma}(\mathcal L, g_1^t)$ and $\overline{\Gamma}(\mathcal L', g_1^t)$ are also zero if $\deg u_2(x)=n$.
\end{thm}
\begin{proof}
 Let $\lambda^2+(a+d)\lambda+1=(\lambda-\rho)(\lambda-\rho^{-1})$ be reducible over $\Bbb F_q$, then by Proposition \ref{prop g},
$$g_1(x)^s=(x+\frac{a+\rho}c)^s,~~~ g_2(x)^t=(x+\frac{a+\rho^{-1}}c)^t \in \Bbb F_q[x]$$
satisfy the condition in  (\ref{eq 2.1}). Hence by Lemma \ref{lem cyclic}, the expurgated Goppa codes $\widetilde{\Gamma}(\mathcal L, g_1^s)$, $\widetilde{\Gamma}(\mathcal L, g_2^t)$ and the extended Goppa codes $\overline{\Gamma}(\mathcal L', g_1^s)$, $\overline{\Gamma}(\mathcal L', g_2^t)$ are all cyclic of length $n$.
 In the following, we only need to find the generator polynomial of $\widetilde{\Gamma}(\mathcal L, g_1^s)$. Similarly, we can also obtain those of $\overline{\Gamma}(\mathcal L', g_1^s)$, $\widetilde{\Gamma}(\mathcal L, g_2^t)$ and $\overline{\Gamma}(\mathcal L', g_2^t)$.

For the cyclic expurgated Goppa code $\widetilde{\Gamma}(\mathcal L, g_1^s)$, we shall prove that the polynomial $u_1(x)$ in (\ref{eq u_1}) is its generator polynomial. Now we consider the parity-check matrix of $\widetilde{\Gamma}(\mathcal L, g_1^s)$ as follows:
$$\widetilde H=
\left(\begin{array}{cccc}
\frac 1{(\alpha_1+\frac{a+\rho}c)^s} & \frac 1{(\alpha_2+\frac{a+\rho}c)^s}& \cdots & \frac 1{(\alpha_n+\frac{a+\rho}c)^s} \\
\frac {\alpha_1}{(\alpha_1+\frac{a+\rho}c)^s} & \frac {\alpha_2}{(\alpha_2+\frac{a+\rho}c)^s} & \cdots & \frac {\alpha_n}{(\alpha_n+\frac{a+\rho}c)^s} \\
\vdots & \vdots &   &  \vdots \\
\frac {\alpha_1^s}{(\alpha_1+\frac{a+\rho}c)^s} & \frac {\alpha_2^s}{(\alpha_2+\frac{a+\rho}c)^s} & \cdots & \frac {\alpha_n^s}{(\alpha_n+\frac{a+\rho}c)^s}
\end{array}\right).
$$
Since $1,x,\ldots,x^s$ and $1,x+\frac{a+\rho}c,\ldots,(x+\frac{a+\rho}c)^s$ are two sets of bases of $\Bbb F_q[x]_{\leq s}$ over $\Bbb F_q$, performing elementary row operations over $\Bbb F_q$ on the above matrix, we can obtain the equivalent matrix over $\Bbb F_q$:
$$
\widetilde H_1=\left(\begin{array}{cccc}
1 & 1 & \ldots & 1 \\
\frac 1{\alpha_1+\frac{a+\rho}c} & \frac 1{\alpha_2+\frac{a+\rho}c} & \cdots & \frac 1{\alpha_n+\frac{a+\rho}c} \\
\vdots & \vdots &   &  \vdots \\
\frac 1{(\alpha_1+\frac{a+\rho}c)^{s-1}} & \frac 1{(\alpha_2+\frac{a+\rho}c)^{s-1}} & \cdots & \frac 1{(\alpha_n+\frac{a+\rho}c)^{s-1}} \\
\frac 1{(\alpha_1+\frac{a+\rho}c)^s} & \frac 1{(\alpha_2+\frac{a+\rho}c)^s}& \cdots & \frac 1{(\alpha_n+\frac{a+\rho}c)^s}
\end{array}\right).
$$
Then $\widetilde H_1$ is also the parity-check matrix of $\widetilde{\Gamma}(\mathcal L, g_1^s)$.

Hence $\boldsymbol c=(c_1,\ldots,c_n) \in \widetilde{\Gamma}(\mathcal L, g_1^s)$ if and only if $\widetilde H_1 \boldsymbol c^\mathrm T=0$ if and only if $\boldsymbol c$ is a solution of the  system of  equations
\begin{equation} \label{Th2,eq1}
\left\{\begin{array}{ll}
&c_1+c_2+\cdots+c_n=0. \\
&\frac{c_1}{\alpha_1+\frac{a+\rho}c}+\cdots+\frac{c_n}{\alpha_n+\frac{a+\rho}c}=0, \\
& \cdots \\
&\frac{c_1}{(\alpha_1+\frac{a+\rho}c)^{s-1}}+\cdots+\frac{c_n}{(\alpha_n+
\frac{a+\rho}c)^{s-1}}=0, \\
&\frac{c_1}{(\alpha_1+\frac{a+\rho}c)^s}+\cdots+\frac{c_n}{(\alpha_n+\frac{a+\rho}c)^s}
=0.
\end{array}\right.
\end{equation}

Let $B=\left(\begin{array}{cc}
0 & 1 \\
1 & \frac{a+\rho}c
\end{array}\right)$
and
$$
D=BAB^{-1}=\left(\begin{array}{cc}
a+d+\rho & c \\
\frac{\rho^2+(a+d)\rho+1}c & \rho
\end{array}\right)=\left(\begin{array}{cc}
\rho^{-2} & c\rho^{-1} \\
0 & 1
\end{array}\right)\triangleq \left(\begin{array}{cc}
e & f \\
0 & 1
\end{array}\right),
$$
where $e=\rho^{-2}$ and $f=c\rho^{-1}$. Then $\ord(D)=\ord(A)=\ord(e)=n$ and
$$
D^i=BA^iB^{-1}=\left(\begin{array}{cc}
e^i & f\frac{e^i-1}{e-1} \\
0 & 1
\end{array}\right), 1\leq i \leq n-1.
$$

In (\ref{Th2,eq1}), $\frac 1{\alpha_i+\frac{a+\rho}c}=B(\alpha_i)=B(A^{i-1}\alpha_1)=BA^{i-1}B^{-1}B(\alpha_1)
=D^{i-1}B(\alpha_1)$, $1\le i\le n$. Denote $\beta=B(\alpha_1)=\frac 1{\alpha_1+\frac{a+\rho}c}$.

Hence the system (\ref{Th2,eq1}) is equivalent to the system
\begin{equation} \label{Th2,eq1'}
\left\{\begin{array}{ll}
&c_1+c_2+\cdots+c_n=0. \\
&c_1\beta+c_2D(\beta)+\cdots+c_nD^{n-1}(\beta)=0, \\
& \cdots \\
&c_1\beta^{s-1}+c_2(D(\beta))^{s-1}+\cdots+c_n(D^{n-1}(\beta))^{s-1}=0, \\
&c_1\beta^s+c_2(D(\beta))^s+\cdots+c_n(D^{n-1}(\beta))^s=0,
\end{array}\right.
\end{equation}
where $\beta=B(\alpha_1)=\frac 1{\alpha_1+\frac{a+\rho}c}$.

We shall prove  that $\boldsymbol c=(c_1,\ldots,c_n) \in \widetilde{\Gamma}(\mathcal L, g_1^s)$ if and only if $\boldsymbol c$ is a solution of the system of equations
\begin{equation} \label{Th2,eq2}
\left\{\begin{array}{ll}
&c_1+c_2+\cdots+c_n=0, \\
&c_1+c_2e+\cdots+c_ne^{n-1}=0, \\
& \cdots \\
&c_1+c_2e^{s-1}+\cdots+c_n(e^{s-1})^{n-1}=0, \\
&c_1+c_2e^s+\cdots+c_n(e^s)^{n-1}=0,
\end{array}\right.
\end{equation}
 where $e=\rho^{-2}$. It means that the system in (\ref{Th2,eq1'}) is equivalent to the system in (\ref{Th2,eq2}). In the following, we shall prove that the system in (\ref{Th2,eq1'}) is equivalent to the system in (\ref{Th2,eq2}) by induction on $s=\deg g(x)$.

If $\deg g(x)=s=1$, then it  has been proved in Theorem \ref{Th 1}. We assume that  $s>1$ and that the statement is shown for $\deg g(x)$ equal to $s-1$, i.e.,
the system
\begin{equation}
\left\{\begin{array}{ll}
&c_1+c_2+\cdots+c_n=0. \\
&c_1\beta+c_2D(\beta)+\cdots+c_nD^{n-1}(\beta)=0, \\
& \cdots \\
&c_1\beta^{s-1}+c_2(D(\beta))^{s-1}+\cdots+c_n(D^{n-1}(\beta))^{s-1}=0,
\end{array}\right.
\end{equation} is equivalent to the system
\begin{equation}
\left\{\begin{array}{ll}
&c_1+c_2+\cdots+c_n=0, \\
&c_1+c_2e+\cdots+c_ne^{n-1}=0, \\
& \cdots \\
&c_1+c_2e^{s-1}+\cdots+c_n(e^{s-1})^{n-1}=0.
\end{array}\right.
\end{equation}
 Now we consider the case $\deg g(x)=s$.

Let  $s=s_1+1$, i.e., $s_1=s-1$, then the last equation in (\ref{Th2,eq1'}) is that
$$
\begin{aligned}
&0=c_1\beta^s+c_2(D(\beta))^s+\cdots+c_n(D^{n-1}(\beta))^s  \\
&=c_1\beta^{s_1}\beta+c_2(D(\beta))^{s_1}(e\beta+f)+\cdots+c_n(D^{n-1}(\beta))^{s_1}
(e^{n-1}\beta+
f\frac{e^{n-1}-1}{e-1})  \\
&=(c_1\beta^{s_1}+c_2(D(\beta))^{s_1} e +\cdots+ c_n(D^{n-1}(\beta))^{s_1} e^{n-1})\beta  \\
&+\frac f{e-1}(c_2(D(\beta))^{s_1} (e-1)+\cdots+c_n(D^{n-1}(\beta))^{s_1} (e^{n-1}-1))  \\
&=(c_1\beta^{s_1}+c_2(D(\beta))^{s_1} e+\cdots+c_n(D^{n-1}(\beta))^{s_1} e^{n-1})\beta  \\
&+\frac f{e-1}(c_2(D(\beta))^{s_1} e+\cdots+c_n(D^{n-1}(\beta))^{s_1} e^{n-1}-c_2(D(\beta))^{s_1}-\cdots-c_n(D^{n-1}(\beta))^{s_1})\\
&=(c_1\beta^{s_1}+c_2(D(\beta))^{s_1} e +\cdots+ c_n(D^{n-1}(\beta))^{s_1} e^{n-1})(\beta+\frac f{e-1}).
\end{aligned}
$$
Since $\beta+\frac f{e-1}=B(\alpha_1)+\frac f{e-1}\ne 0$ (otherwise, $\frac 1{\alpha_1+\frac{a+\rho}c} = \frac{c\rho^{-1}}{\rho^{-2}-1}$, i.e.,   $\alpha_1=\frac{a+\rho^{-1}}c$, which is contradictory with the  choice of $\alpha_1$),  the last equation in (\ref{Th2,eq1'}) is equivalent to
\begin{align}
&0=c_1\beta^{s_1}+c_2(D(\beta))^{s_1} e +\cdots+ c_n(D^{n-1}(\beta))^{s_1} e^{n-1}\nonumber\\
&=c_1\beta^{s_1}+c_2(e\beta+f\frac {e-1}{e-1})^{s_1} e+\cdots+c_n(e^{n-1}\beta+f\frac{e^{n-1}-1}{e-1})^{s_1} e^{n-1}\nonumber\\
&=(c_1+c_2e^s+\cdots+c_n(e^s)^{n-1})\beta^{s_1}\nonumber\\
&+\sum\limits_{i=1}^{s_1} \frac{C_{s_1}^i f^i \beta^{s_1-i}}{(e-1)^i}(c_2e^{s_1+1-i}(e-1)^i+\cdots+
c_n(e^{n-1})^{s_1+1-i}(e^{n-1}-1)^i),
\end{align}
where each $C_{s_1}^i$ is a combinatorial number.

For $1\leq i\leq s_1$,  $i=\sum\limits_{j=1}^k 2^{d_j}$ is called the $2-adic$ expansion of $i$, where $d_1<d_2<\cdots <d_k$, and $k$ is called the algebraic degree of $i$.  In (3.14),
\begin{eqnarray*}
&&c_2e^{s_1+1-i}(e-1)^i+\cdots+c_n(e^{n-1})^{s_1+1-i}(e^{n-1}-1)^i \\
&&=c_2e^{s-i}\prod\limits_{j=1}^k(e^{2^{d_j}}+1)+\cdots+c_n(e^{n-1})^{s-i}
\prod\limits_{j=1}^k((e^{n-1})^{2^{d_j}}+1)\\
&&=c_2e^s+c_3(e^s)^2+\cdots +c_n(e^{s})^{n-1}\\
&&+\sum_{j=1}^{2^k-1}(c_2e^{s-i+t_j}+c_3(e^{s-i+t_j})^2+\cdots+c_n(e^{s-i+t_j})^{n-1}),
\end{eqnarray*}
where each $t_j$ is a partial sum of $2^{d_1}, 2^{d_2},\ldots, 2^{d_k}$ and $t_j<i$. By the induction hypothesis, each  $c_2 e^{s-i+t_j}+c_3(e^{s-i+t_j})^2+\cdots+c_n (e^{s-i+t_j})^{n-1}=c_1$.  Hence the equation in (3.14) is equivalent to
\begin{eqnarray*}
&(c_1+c_2e^s+\cdots+c_n(e^s)^{n-1})(\beta^{s_1}+\sum\limits_{i=1}^{s_1} \frac{C_{s_1}^i f^i \beta^{s_1-i}}{(e-1)^i})\\
&=(c_1+c_2e^s+\cdots+c_n(e^s)^{n-1})(\beta+\frac f{e-1})^{s_1}=0,
\end{eqnarray*}
where $\beta=B(\alpha_1)=\frac 1{\alpha_1+\frac {a+\rho}c}$ and $\beta+\frac f{e-1}\ne 0$. Then the last equation in (\ref{Th2,eq1'}) holds if and only if $c_1+c_2e^s+\cdots+c_n(e^s)^{n-1}=0$ under the induction hypothesis, which means the system in (\ref{Th2,eq1'}) is equivalent to the system in (\ref{Th2,eq2}).

Therefore, $\boldsymbol c=(c_1,\ldots,c_n) \in \widetilde{\Gamma}(\mathcal L, g_1^s)$ if and only if $\boldsymbol c$ is a solution of the system (\ref{Th2,eq2}).
 Moreover, let $m_{\rho^{-i}}(x)$ be the minimal polynomials of $\rho^{-i}$ over $\Bbb F_2$, then they are also the minimal polynomials of  $e^i=\rho^{-2i}$ over $\Bbb F_2$, $i=1,2,\ldots, s$. Since $\widetilde{\Gamma}(\mathcal L, g_1^s)$ is a cyclic code of length $n$, whose elements satisfy the system (\ref{Th2,eq2}), its generator polynomial is $$u_1(x)=(x+1)\mbox{lcm}\{m_{\rho^{-i}}(x):i=1,2,\ldots, s\}.$$

  Similarly, we can get generator polynomial of $\overline{\Gamma}(\mathcal L', g_1^s)$ is $u_1(x)$ and generator polynomials of both $\widetilde{\Gamma}(\mathcal L, g_2^t)$ and $\overline{\Gamma}(\mathcal L', g_2^t)$ are $$u_2(x)=(x+1)\mbox{lcm}\{m_{\rho^i}(x): i=1,2,\ldots, t\},$$ where $m_{\rho^i}(x)$ are the minimal polynomials of $\rho^i$ over $\Bbb F_2$.

  If $\deg u_1(x)=n$, then the dimensions of both $\widetilde{\Gamma}(\mathcal L, g_1^s)$ and $\overline{\Gamma}(\mathcal L', g_1^s)$ are $n-\deg u_1(x)=0$, so both $\widetilde{\Gamma}(\mathcal L, g_1^s)$ and $\overline{\Gamma}(\mathcal L', g_1^s)$ are zero; similarly, both $\widetilde{\Gamma}(\mathcal L, g_1^t)$ and $\overline{\Gamma}(\mathcal L', g_1^t)$ are also zero if $\deg u_2(x)=n$.
\end{proof}

\begin{rem}\label{rb}
In Theorem \ref{Th 2}, since $\rho^j$ and $\rho^{2j}$ have the same minimal polynomial over $\Bbb F_2$ for every $j\in \Bbb Z$, (\ref{eq u_1}) and (\ref{eq u_2}) in Theorem \ref{Th 2} can be changed into
$$
\begin{aligned}
u_1(x)&=(x+1)\mbox{lcm}\{m_{\rho^{-i}}(x): i=1,2,\ldots, 2\lfloor\frac{s+1}2\rfloor\}\\
&=(x+1)\mbox{lcm}\{m_{\rho^{-i}}(x): i=1,3,\ldots, 2\lfloor\frac{s+1}2\rfloor-1\}
\end{aligned}
$$
and
$$
\begin{aligned}
u_2(x)&=(x+1)\mbox{lcm}\{m_{\rho^i}(x): i=1,2,\ldots, 2\lfloor\frac{t+1}2\rfloor\}\\
&=(x+1)\mbox{lcm}\{m_{\rho^i}(x): i=1,3,\ldots, 2\lfloor\frac{t+1}2\rfloor-1\}.
\end{aligned}
$$
\end{rem}

By Definition \ref{BCH}, Theorem \ref{Th 2}, and Remark \ref{rb}, we can immediately obtain the following result.

\begin{cor}
(1) In Theorem \ref{Th 2}, the expurgated Goppa code $\widetilde{\Gamma}(\mathcal L, g_1^s)$ and the extended Goppa code $\overline{\Gamma}(\mathcal L', g_1^s)$ are BCH codes of length $n$ with minimum distances
$d\ge 2\lfloor\frac{s+1}2\rfloor+2$.

(2) In Theorem \ref{Th 2}, the expurgated Goppa code $\widetilde{\Gamma}(\mathcal L, g_1^t)$ and the extended Goppa code $\overline{\Gamma}(\mathcal L', g_1^t)$ are BCH codes of length $n$ with minimum distances
$d\ge 2\lfloor\frac{t+1}2\rfloor+2$.
\end{cor}
\begin{cor}\label{Th2,cor,irr}
Let $A\in PGL_2(\Bbb F_q)$ be the matrix, $G=\langle A \rangle$ the subgroup in Theorem \ref{Th 2}, and $|\lambda E_2-A|=\lambda^2+(a+d)\lambda+1=(\lambda-\rho)(\lambda-\rho^{-1})$ irreducible over $\Bbb F_q$. Let
$$
\mathcal L'=O_{\infty}\backslash \{\infty\}=\{ P(\rho ),\ldots, P(\rho^{n-1} )\}=\{\gamma, A(\gamma), \ldots, A^{n-2}(\gamma)\} \subseteq \Bbb F_{q^2}
$$
and
$$\mathcal L=O_{\alpha}=\{\alpha_1=\alpha, \alpha_2=A(\alpha),\ldots,\alpha_n=A^{n-1}(\alpha)\} \subseteq \Bbb F_{q^2}$$
be the orbit of $\alpha$ under the action of $G$ if $\alpha \notin \{\frac{a+\rho}c,\frac{a+\rho^{-1}}c\}\cup O_{\infty}$.

 Then the conclusions of Theorem \ref{Th 2} still hold but with the Goppa polynomials $g_1^s(x), g_2^t(x)\in \Bbb F_{q^2}[x]$.
 Moreover, let $\lambda^2+(a+d)\lambda+1$ be irreducible over $\Bbb F_q$, $\ord(A)=n\mid (q+1)$, then by Remark \ref{WY}, $m_{\rho^i}(x)=m_{\rho^{-i}}(x),i=1,2,\ldots, j, j=\max \{s,t\}.$
\end{cor}

\begin{exa}
In Example \ref{eg 1}, the minimal polynomials of $\rho^{-3}$, $\rho^3$, $\rho^{-5}$, $\rho^5$, $\rho^{-7}$, and $\rho^7$ over $\Bbb F_2$ are $m_{\rho^{-3}}(x)=x^3+x^2+1$, $m_{\rho^3}(x)=x^3+x+1$, $m_{\rho^{-5}}(x)=m_\rho(x)=x^6+x^5+x^4+x^2+1$, $m_{\rho^5}(x)=m_{\rho^{-1}}(x)=x^6+x^4+x^2+x+1$, $m_{\rho^{-7}}(x)=m_{\rho^7}(x)=x^2+x+1$, respectively. Let $g_1(x)=x+\frac{\alpha^5+\alpha^3}{\alpha^{13}}\in \Bbb F_{2^6}[x]$, $g_2(x)=x+\frac{\alpha^5+\alpha^{-3}}{\alpha^{13}} \in \Bbb F_{2^6}[x]$.

(1)\begin{eqnarray*}
&&\widetilde{\Gamma}(\mathcal L, g_1^3)=\widetilde{\Gamma}(\mathcal L, g_1^4)=\overline{\Gamma}(\mathcal L', g_1^3)=\overline{\Gamma}(\mathcal L', g_1^4)\\
&&=\langle(x+1)(x^3+x^2+1)(x^6+x^4+x^2+x+1)\rangle\\
&&=\langle(x+1)\mbox{lcm}(m_{\rho^{-1}}(x),m_{\rho^{-3}}(x))\rangle,
\end{eqnarray*}
\begin{eqnarray*}
&&\widetilde{\Gamma}(\mathcal L, g_2^3)=\widetilde{\Gamma}(\mathcal L, g_2^4)=\overline{\Gamma}(\mathcal L', g_2^3)=\overline{\Gamma}(\mathcal L', g_2^4)\\
&&=\langle(x+1)(x^3+x+1)(x^6+x^5+x^4+x^2+x+1)\rangle\\
&&=\langle(x+1)\mbox{lcm}(m_{\rho}(x),m_{\rho^3}(x))\rangle,
\end{eqnarray*}
and they are all $[21, 11, 6]$ cyclic codes by the Magma program.

(2) \begin{eqnarray*}
&&\widetilde{\Gamma}(\mathcal L, g_1^5)=\widetilde{\Gamma}(\mathcal L, g_1^6)=\overline{\Gamma}(\mathcal L', g_1^5)=\overline{\Gamma}(\mathcal L', g_1^6)\\
&&=\langle(x+1)(x^3+x^2+1)(x^6+x^4+x^2+x+1)(x^6+x^5+x^4+x^2+1)\rangle\\
&&=\langle(x+1)\mbox{lcm}\{m_{\rho^{-i}}(x):i=1,3,5\}\rangle,
\end{eqnarray*}
\begin{eqnarray*}
&&\widetilde{\Gamma}(\mathcal L, g_2^5)=\widetilde{\Gamma}(\mathcal L, g_2^6)=\overline{\Gamma}(\mathcal L', g_2^5)=\overline{\Gamma}(\mathcal L', g_2^6)\\
&&=\langle(x+1)(x^3+x+1)(x^6+x^4+x^2+x+1)(x^6+x^5+x^4+x^2+1)\rangle\\
&&=\langle(x+1)\mbox{lcm}\{m_{\rho^{i}}(x):i=1,3,5\}\rangle,
\end{eqnarray*}
and they are all $[21, 5, 10]$ cyclic codes by the Magma program.

(3) \begin{eqnarray*}
&&\widetilde{\Gamma}(\mathcal L, g_1^7)=\widetilde{\Gamma}(\mathcal L, g_1^8)=\overline{\Gamma}(\mathcal L', g_1^7)=\overline{\Gamma}(\mathcal L', g_1^8)\\
&&=\langle(x+1)(x^2+x+1)(x^3+x^2+1)(x^6+x^4+x^2+x+1)(x^6+x^5+x^4+x^2+1)\rangle\\
&&=\langle(x+1)\mbox{lcm}\{m_{\rho^{-i}}(x):i=1,3,5,7\}\rangle,
\end{eqnarray*}
\begin{eqnarray*}
&&\widetilde{\Gamma}(\mathcal L, g_2^7)=\widetilde{\Gamma}(\mathcal L, g_2^8)=\overline{\Gamma}(\mathcal L', g_2^7)=\overline{\Gamma}(\mathcal L', g_2^8)\\
&&=\langle(x+1)(x^2+x+1)(x^3+x+1)(x^6+x^4+x^2+x+1)(x^6+x^5+x^4+x^2+1)\rangle\\
&&=\langle(x+1)\mbox{lcm}\{m_{\rho^{i}}(x):i=1,3,5,7\}\rangle,
\end{eqnarray*}
and they are all $[21, 3, 12]$ cyclic codes by the Magma program.
\end{exa}

\subsection{$g(x)=g_1(x)^sg_2(x)^t$, $s,t \in \mathbb{N}$ and $s+ t<n-1$}\

In the subsection, we shall assume that $g(x)=g_1(x)^sg_2(x)^t$, where $s,t \in \mathbb{N}$ and $s+t<n-1$. We shall find generator polynomials of cyclic expurgated or extended Goppa codes in the case of $\sigma^j=1$ of Section 2.3.

\begin{thm} \label{Th3}
Let $A=\left(\begin{array}{cc}
a & b \\
c & d
\end{array}\right) \in PGL_2(\Bbb F_q)$ be of order $n>2$ with $|A|=ad+bc=1$ and $c\ne 0$, $G=\langle A\rangle$ a subgroup of $PGL_2(\Bbb F_q)$, and $|\lambda E_2-A|=\lambda^2+(a+d)\lambda+1=(\lambda-\rho)(\lambda-\rho^{-1})$ reducible over $\Bbb F_q$.
Let
$$
\mathcal L'=O_{\infty}\backslash \{\infty\}=\{ P(\rho),\ldots, P(\rho^{n-1} )\}=\{\gamma, A(\gamma), \ldots, A^{n-2}(\gamma)\} \subseteq \Bbb F_q
$$
and
$$\mathcal L=O_{\alpha}=\{\alpha_1=\alpha, \alpha_2=A(\alpha),\ldots,\alpha_n=A^{n-1}(\alpha)\} \subseteq \Bbb F_{q}$$
be the orbit of $\alpha$ under the action of $G$ if $\alpha \notin \{\frac{a+\rho}c,\frac{a+\rho^{-1}}c\}\cup O_{\infty}$.

If $g(x)=g_1(x)^sg_2(x)^t=(x+\frac{a+\rho}c)^s (x+\frac{a+\rho^{-1}}c)^t \in \Bbb F_q[x]$, then both the expurgated Goppa code $\widetilde{\Gamma}(\mathcal L, g)$ and the extended Goppa code $\overline{\Gamma}(\mathcal L', g)$ are cyclic of length $n$ and their generator polynomials are
\begin{equation} \label{eq h}
u(x)=\mbox{lcm}(u_1(x), u_2(x)),
\end{equation}
where $\widetilde{\Gamma}(\mathcal L, g_1^s)=\overline{\Gamma}(\mathcal L', g_1^s)=\langle u_1(x)\rangle$ and $\widetilde{\Gamma}(\mathcal L, g_2^t)=\overline{\Gamma}(\mathcal L', g_2^t)=\langle u_2(x)\rangle$, $u_1(x)$ and $u_2(x)$ are defined as (\ref{eq u_1}) and (\ref{eq u_2}).

Moreover, if $\deg u(x)=n$, then both $\widetilde{\Gamma}(\mathcal L, g)$ and $\overline{\Gamma}(\mathcal L', g)$ are zero.
\end{thm}
\begin{proof}
 Let $\lambda^2+(a+d)\lambda+1=(\lambda-\rho)(\lambda-\rho^{-1})$ be reducible over $\Bbb F_q$, then by Proposition \ref{prop g},
$$g(x)=g_1(x)^sg_2(x)^t \in \Bbb F_q[x]$$
satisfies the condition in (\ref{eq 2.1}). Hence by Lemma \ref{lem cyclic}, both the expurgated Goppa code $\widetilde{\Gamma}(\mathcal L, g)$ and the extended Goppa code $\overline{\Gamma}(\mathcal L', g)$ are cyclic of length $n$.

Furthermore, by $\gcd(g_1(x),g_2(x))=1$, $\Gamma(\mathcal L, g)=\Gamma(\mathcal L, g_1^s) \cap \Gamma(\mathcal L, g_2^t)$, then
$$
\widetilde{\Gamma}(\mathcal L,g)=\widetilde{\Gamma}(\mathcal L, g_1^sg_2^t)=\widetilde{\Gamma}(\mathcal L, g_1^s) \cap \widetilde{\Gamma}(\mathcal L, g_2^t)=\langle \mbox{lcm}(u_1(x), u_2(x))\rangle,
$$
 $$
  \overline{\Gamma}(\mathcal L',g)=\overline{\Gamma}(\mathcal L', g_1^sg_2^t)=\overline{\Gamma}(\mathcal L', g_1^s) \cap \overline{\Gamma}(\mathcal L', g_2^t)=\langle \mbox{lcm}(u_1(x), u_2(x))\rangle,
 $$
  where $\widetilde{\Gamma}(\mathcal L, g_1^s)=\overline{\Gamma}(\mathcal L, g_1^s)=\langle u_1(x)\rangle$ and
 $\widetilde{\Gamma}(\mathcal L, g_2^t)=\overline{\Gamma}(\mathcal L', g_2^t)=\langle u_2(x)\rangle.$

 If $\deg u(x)=n$, then the dimensions of both $\widetilde{\Gamma}(\mathcal L, g)$ and $\overline{\Gamma}(\mathcal L', g)$ are $n-\deg u(x)=0$, so both $\widetilde{\Gamma}(\mathcal L, g)$ and $\overline{\Gamma}(\mathcal L', g)$ are zero.
\end{proof}

\begin{cor}\label{Th3,cor,irr}
Let $A\in PGL_2(\Bbb F_q)$ be the matrix, $G=\langle A \rangle$ the subgroup in Theorem \ref{Th3}, and $|\lambda E_2-A|=\lambda^2+(a+d)\lambda+1=(\lambda-\rho)(\lambda-\rho^{-1})$ irreducible over $\Bbb F_q$. Let
$$
\mathcal L'=O_{\infty}\backslash \{\infty\}=\{ P(\rho ),\ldots, P(\rho^{n-1} )\}=\{\gamma, A(\gamma), \ldots, A^{n-2}(\gamma)\} \subseteq \Bbb F_{q^2}
$$
and
$$\mathcal L=O_{\alpha}=\{\alpha_1=\alpha, \alpha_2=A(\alpha),\ldots,\alpha_n=A^{n-1}(\alpha)\} \subseteq \Bbb F_{q^2}$$
be the orbit of $\alpha$ under the action of $G$ if $\alpha \notin \{\frac{a+\rho}c,\frac{a+\rho^{-1}}c\}\cup O_{\infty}$.

 Then the conclusions of Theorem \ref{Th3} still hold but with the Goppa polynomials $g(x)=g_1^s(x)g_2^t(x)\in \Bbb F_{q^2}[x]$.
 Moreover, let $\lambda^2+(a+d)\lambda+1$ be irreducible over $\Bbb F_q$, $\ord(A)=n\mid (q+1)$, then by Remark \ref{WY}, $m_{\rho^i}(x)=m_{\rho^{-i}}(x),i=1,2,\ldots, j, j=\max \{s,t\}.$
\end{cor}

\begin{rem}
In Theorem \ref{Th3}, let $\lambda^2+(a+d)\lambda+1=(\lambda-\rho)(\lambda-\rho^{-1})$ be irreducible over $\Bbb F_q$, $g(x)=(g_1(x)g_2(x))^s \in \Bbb F_q[x]$ a monic polynomial satisfying the condition in (\ref{eq 2.1}), and $\mathcal L, \mathcal L' \subseteq \Bbb F_q$. Then the conclusions of Theorem \ref{Th3} still hold.

So far, we have given all the answers to the case of $\sigma^j=1$ of Section 2.3.
\end{rem}

\begin{exa} In Example \ref{eg 1}, the minimal polynomials of $\rho^{-3}$, $\rho^3$ $\rho^{-5}$, $\rho^5$, $\rho^{-7}$, and $\rho^7$ over $\Bbb F_2$ are $m_{\rho^{-3}}(x)=x^3+x^2+1$, $m_{\rho^3}(x)=x^3+x+1$, $m_{\rho^{-5}}(x)=m_\rho(x)=x^6+x^5+x^4+x^2+1$, $m_{\rho^5}(x)=m_{\rho^{-1}}(x)=x^6+x^4+x^2+x+1$, $m_{\rho^{-7}}(x)=m_{\rho^7}(x)=x^2+x+1$, respectively. Let $g_1(x)=x+\frac{\alpha^5+\alpha^3}{\alpha^{13}}\in \Bbb F_{2^6}[x]$, $g_2(x)=x+\frac{\alpha^5+\alpha^{-3}}{\alpha^{13}} \in \Bbb F_{2^6}[x]$.

(1) If $g(x)=g_1(x)g_2(x)$, then
\begin{eqnarray*}
&&\widetilde{\Gamma}(\mathcal L, g)=\overline{\Gamma}(\mathcal L', g)\\
&&=\langle(x+1)(x^6+x^4+x^2+x+1)(x^6+x^5+x^4+x^2+x+1)\rangle\\
&&=\langle(x+1)\mbox{lcm}(m_{\rho^{-1}}(x),m_{\rho}(x))\rangle,
\end{eqnarray*}
and they are $[21,8,6]$ cyclic codes by the Magma program.

(2) If $g(x)=g_1(x)g_2(x)^3$ or $g(x)=g_1(x)g_2(x)^5$, then
\begin{eqnarray*}
&&\widetilde{\Gamma}(\mathcal L, g)=\overline{\Gamma}(\mathcal L', g)\\
&&=\langle(x+1)(x^3+x+1)(x^6+x^4+x^2+x+1)(x^6+x^5+x^4+x^2+x+1)\rangle\\
&&=\langle(x+1)\mbox{lcm}(m_{\rho^{-1}}(x),m_{\rho}(x),m_{\rho^3}(x))\rangle,
\end{eqnarray*}
and they are $[21, 5, 10]$ cyclic codes by the Magma program.

(3) If $g(x)=g_1(x)^sg_2(x)^t$, where $(s,t)\in \{(3,3),(3,5),(5,3),(5,5)\}$, then
\begin{eqnarray*}
&&\widetilde{\Gamma}(\mathcal L, g)=\overline{\Gamma}(\mathcal L', g)\\
&&=\langle(x+1)(x^3+x+1)(x^3+x^2+1)(x^6+x^4+x^2+x+1)(x^6+x^5+x^4+x^2+x+1)\rangle\\
&&=\langle(x+1)\mbox{lcm}(m_{\rho^{-1}}(x),m_{\rho^{-3}}(x),m_{\rho}(x),
m_{\rho^3}(x))\rangle,
\end{eqnarray*}
and they are $[21, 2, 14]$ cyclic codes by the Magma program.

(4) If $g(x)=g_1(x)g_2(x)^7$, then
\begin{eqnarray*}
&&\widetilde{\Gamma}(\mathcal L, g)=\overline{\Gamma}(\mathcal L', g)\\
&&=\langle(x+1)(x^2+x+1)(x^3+x+1)(x^6+x^4+x^2+x+1)(x^6+x^5+x^4+x^2+x+1)\rangle\\
&&=\langle(x+1)\mbox{lcm}(m_{\rho^{-1}}(x),m_{\rho}(x),m_{\rho^3}(x),m_{\rho^5}(x),
m_{\rho^7}(x))\rangle,
\end{eqnarray*}
and they are $[21, 3, 12]$ cyclic codes by the Magma program.

(5) If $g(x)=g_1(x)^7g_2(x)$, then
\begin{eqnarray*}
&&\widetilde{\Gamma}(\mathcal L, g)=\overline{\Gamma}(\mathcal L', g)\\
&&=\langle(x+1)(x^2+x+1)(x^3+x^2+1)(x^6+x^4+x^2+x+1)(x^6+x^5+x^4+x^2+x+1)\rangle\\
&&=\langle(x+1)\mbox{lcm}(m_{\rho^{-1}}(x),m_{\rho^{-3}}(x),m_{\rho^{-5}}(x),
m_{\rho^{-7}}(x),m_{\rho}(x))\rangle,
\end{eqnarray*}
and they are $[21, 3, 12]$ cyclic codes by the Magma program.
\end{exa}

Let $M=\left(\left(\begin{array}{cc}
a & b \\
c & d
\end{array}\right),\sigma^j\right) \in P\Gamma L_2(\Bbb F_q)$ and $g(x)=\sum\limits_{i=0}^r g_i x^i \in\Bbb F_q[x]$ a monic polynomial of degree $r$ satisfying the condition in (\ref{eq 2.1}). We gave generator polynomials of cyclic expurgated or extended Goppa codes if $\sigma^j=1$ in this paper. Based on the above discussion, now we have an open question:

If $\sigma^j\ne 1$,  can you show generator polynomials of cyclic expurgated or extended Goppa codes?

\section{Conclusion}
Cyclic codes are an interesting type of linear codes and have wide applications in communication and storage systems due to their efficient encoding and decoding algorithms. Goppa codes are widely used in public-key cryptosystems. Cyclic Goppa codes, cyclic expurgated and extended Goppa codes were  investigated (e.g., \cite{Ba,Bb,Bc,Bd,Be, LY}), but their generator polynomials have not shown so far.

  This study aims to determine all the generator polynomials of cyclic expurgated or extended Goppa codes under some prescribed permutations induced by the projective general linear automorphism $A \in PGL_2(\Bbb F_q)$.
  Future research should be conducted to investigate the generator polynomials of cyclic expurgated or extended Goppa codes under some prescribed permutations induced by the projective semi-linear automorphism $M \in P\Gamma L_2(\Bbb F_q)$.

  \section*{Acknowledgement}
The authors would like to thank the editor and anonymous reviewers for their valuable comments that improved the quality of this paper.

\end{document}